\documentclass[acmsmall,nonacm]{acmart} 

\AtBeginDocument{%
	\providecommand\BibTeX{{%
			\normalfont B\kern-0.5em{\scshape i\kern-0.25em b}\kern-0.8em\TeX}}}

\usepackage{tikz}
\usepackage[linesnumbered,ruled,noend]{algorithm2e}

\usetikzlibrary{positioning,calc}
\usetikzlibrary{decorations.text}
\usetikzlibrary{decorations.pathmorphing,decorations.pathreplacing}
\usetikzlibrary{arrows,petri, topaths,fit}

\usepackage{listings}
\usepackage{wrapfig}
\usepackage{framed}   
 {\endMakeFramed}
\definecolor{shadecolor}{gray}{0.75}
\usepackage{thmtools, thm-restate}

\usepackage{multirow}
\usepackage{subcaption}
\usepackage{url}
\usepackage{graphicx}
\usepackage{tcolorbox}
\usepackage{enumitem}
\usepackage{array}
\usepackage{makecell}

\definecolor{mycolor}{rgb}{0.122, 0.435, 0.698}
\makeatletter
\newcommand{\mybox}[1]{%
  \setbox0=\hbox{#1}%
  \setlength{\@tempdima}{\dimexpr\wd0+13pt}%
  \begin{tcolorbox}[colframe=mycolor,boxrule=0.5pt,arc=4pt,
      left=6pt,right=6pt,top=3pt,bottom=3pt,boxsep=0pt,width=\@tempdima]
    #1
  \end{tcolorbox}
}

\newcommand{\introparagraph}[1]{\vspace{0.7mm} \noindent \textbf{\em #1.}}
\newcommand{\relatedwork}[1]{\vspace{0.7mm} \noindent \textbf{#1.}}

\newcommand{\eat}[1]{}

\newcommand{\ba}{\mathbf{c}}

\newcommand{\obtainedfrom}{ \leftarrow }

\newcommand{\eqone}{\stackrel{\text{(1)}}{=}}
\newcommand{\eqtwo}{\stackrel{\text{(2)}}{=}}
\newcommand{\eqthree}{\stackrel{\text{(3)}}{=}}
\newcommand{\defeq}{:=}

\newcommand{\fhw}{\mathsf{fhw}}
\newcommand{\subw}{\mathsf{subw}}
\newcommand{\vars}[1]{\mathsf{vars}(#1)}
\newcommand{\cfhw}{\mathsf{ffhw}}
\newcommand{\csubw}{\mathsf{fsubw}} 

\newcommand{\bL}{\mathbf{L}}
\newcommand{\bS}{\sigma}
\newcommand{\zerobf}{\mathbf{0}}
\newcommand{\onebf}{\mathbf{1}}

\newcommand{\nodes}{[\ell]}
\newcommand{\mE}{\mathcal{E}}

\newcommand{\bx}{\mathbf{x}}

\newcommand{\Datalogo}{\textsf{Datalog}\xspace}
\newcommand{\Datalog}{\textsf{Datalog}\xspace}

\newcommand{\SumProd}{\textsf{sum-prod}\xspace}
\newcommand{\sumprod}{\varphi}
\newcommand{\POPS}{\textsf{POPS}\xspace}

\newcommand{\arity}[1]{\mathsf{arity}(#1)}

\newcommand{\IDB}{\textsf{IDB}\xspace}
\newcommand{\EDB}{\textsf{EDB}\xspace}

\newcommand{\IDBs}{\textsf{IDB}s\xspace}
\newcommand{\EDBs}{\textsf{EDB}s\xspace}

\newcommand{\mT}{\mathcal{T}}

\newcommand{\revone}[1]{{\color{black} #1}}
\definecolor{reviewer2-color}{rgb}{0.8,0.0, 0.9}
\newcommand{\revtwo}[1]{{\color{black} #1}}
\newcommand{\revthree}[1]{{\color{black} #1}}

\def\isarxiv{1}

\if\isarxiv1
\else
\setcopyright{acmlicensed}
\acmJournal{PACMMOD}
\acmYear{2024} \acmVolume{2} \acmNumber{N2 (PODS)}
\acmArticle{XXX} \acmMonth{5} \acmPrice{15.00}
\acmDOI{10.1145/XXXXXXX}
\fi

\author{Hangdong Zhao}
\affiliation{%
  \institution{University of Wisconsin-Madison}
  \city{Madison}
  \country{USA}}
\email{hangdong@cs.wisc.edu}

\author{Shaleen Deep}
\affiliation{%
  \institution{Microsoft Gray Systems Lab}
  \country{USA}}
\email{shaleen.deep@microsoft.com}

\author{Paraschos Koutris}
\affiliation{%
  \institution{University of Wisconsin-Madison}
  \city{Madison}
  \country{USA}}
\email{paris@cs.wisc.edu}

\author{Sudeepa Roy}
\affiliation{%
  \institution{Duke University}
  \city{Durham}
  \country{USA}}
\email{sudeepa@cs.duke.edu}

\author{Val Tannen}
\affiliation{%
  \institution{University of Pennsylvania}
  \city{Philadelphia}
  \country{USA}}
\email{val@seas.upenn.edu}

\renewcommand{\shortauthors}{Zhao et al.}
\begin{document}

\title{Evaluating \Datalog over Semirings: A Grounding-based Approach}

\begin{abstract}
	\eat{
		\Datalogo \cite{POPS} extends \textsf{Datalog}, where each rule of the program is a sum of sum-product queries over any commutative partially ordered pre-semiring (\POPS), instead of the standard Boolean semiring. It can express recursive aggregation beyond the Booleans, while retaining the declarative fixpoint semantics: the naive evaluation algorithm can find the least fixpoint (if exists) by iteratively applying the immediate consequence operator until the program quiesces.  

		It was shown in \cite{POPS} that, the convergence can be studied through the core of the underlying \POPS, which is a naturally-ordered semiring. In this work, we ask: {\em given a \Datalogo program over a naturally-ordered semiring $\bS$, what is the tightest possible runtime?} To this end, our main contribution is a general two-phase framework for analyzing the data complexity of \Datalogo: first ground the program into an equivalent system of polynomial equations (i.e. grounding) and then find the least fixpoint of the grounding over $\bS$. We present algorithms that use structure-aware query evaluation techniques to obtain the smallest possible groundings. Next, efficient algorithms for fixpoint evaluation are introduced over two classes of semirings: (1) semirings of finite rank, and (2) absorptive semirings of total order. Combining both phases, we obtain state-of-the-art and new algorithmic results. Finally, we complement our results with a matching fine-grained lower bound.
	}

    \Datalog is a powerful yet elegant language that allows expressing recursive computation. Although \Datalog evaluation has been extensively studied in the literature, so far, only loose upper bounds are known on how fast a \Datalog program can be evaluated. In this work, we ask the following question: given a \Datalog program over a naturally-ordered semiring $\bS$, what is the tightest possible runtime? To this end, our main contribution is a general two-phase framework for analyzing the data complexity of \Datalog over $\bS$: first ground the program into an equivalent system of polynomial equations (i.e. grounding) and then find the least fixpoint of the grounding over $\bS$. We present algorithms that use structure-aware query evaluation techniques to obtain the smallest possible groundings. Next, efficient algorithms for fixpoint evaluation are introduced over two classes of semirings: (1) finite-rank semirings and (2) absorptive semirings of total order. Combining both phases, we obtain state-of-the-art and new algorithmic results. Finally, we complement our results with a matching fine-grained lower bound.
    
\end{abstract}



\maketitle
\allowdisplaybreaks
\setcounter{page}{1}
\section{Introduction}
\label{sec:intro}

\Datalog is a recursive query language that has gained prominence due to its expressivity and rich applications across multiple domains, including graph processing~\cite{SeoPSL13}, declarative program analysis~\cite{Souffle,SmaragdakisB15}, and business analytics~\cite{GreenHLZ13}. Most of the prior work focuses on \Datalog programs over the Boolean semiring (this corresponds to the standard relational join semantics): popular examples include same generation~\cite{bancilhon1985magic}, cycle finding, and pattern matching. Many program analysis tasks such as Dyck-reachability~\cite{Reps98, li2020fast, li2021complexity}, context-free reachability~\cite{Andersen}, and Andersen's analysis~\cite{andersen1994program} can also be naturally cast as \Datalog programs over the Boolean domain. However, modern data analytics frequently involve aggregations over recursion. Seminal work by Green et al.~\cite{green2007provenance} established the semantics of \Datalog over the so-called $K$-relations where tuples are annotated by the domain of a fixed semiring. \eat{Prior work has also studied the fundamental problem of computing \Datalog provenance by leveraging the proof-theoretic and fixpoint-theoretic derivations~\cite{RamusatMS21, DynDatalog}, as well as through building compact circuits for various semirings~\cite{deutch2014circuits}.} Recently, Abo Khamis et al.~\cite{POPS} proposed \textsf{Datalog}$^\circ$, an elegant extension of \Datalog, and established key algebraic properties that governs convergence of \textsf{Datalog}$^\circ$. The authors further made the observation that the convergence rate of \Datalog can be studied by confining to the class of {\em naturally-ordered} semirings (formally defined in Section~\ref{sec:prelim}).

A parallel line of work has sought to characterize the complexity of \Datalog evaluation. The general data complexity of \Datalog is \textsf{P}-complete~\cite{garey1997computers, cosmadakis1999inherent}, with the canonical \textsf{P}-complete \Datalog program:
\begin{equation}
\begin{aligned}
	T(x_1) \obtainedfrom U(x_1). \quad \quad T(x_1) \obtainedfrom T(x_2) \wedge T(x_3) \wedge R(x_2, x_3, x_1).
\end{aligned}
\label{eq:program}
\end{equation}

Some fragments of \Datalog can have lower data complexity: the evaluation for non-recursive \Datalog is in $\mathsf{AC}_0$, whereas evaluation for \Datalog with linear rules is in $\mathsf{NC}$ and thus efficiently parallelizable~\cite{UllmanG88, afrati1993parallel}. However, all such results do not tell us how efficiently we can evaluate a given \Datalog program. As an example, the program~\eqref{eq:program} can be evaluated in linear time with respect to the input size $m$~\cite{GK04}. \revtwo{A deeper understanding of precise upper bounds for Datalog is important, given that it can capture practical problems
(all of which are in P) across various domains as mentioned before}. 

Unfortunately, even for the Boolean semiring, the current general algorithmic techniques for \Datalog\ evaluation typically aim for an imprecise polynomial bound instead of specifying the tightest possible exponent. Semi-na\"{i}ve or na\"{i}ve evaluation only provides upper bounds on the number of iterations, ignoring the computational cost of each iteration. For example, program~\eqref{eq:program} can be evaluated in at most $O(m)$ iterations, but the cost of an iteration can be as large as $\Theta(m)$. Further, going beyond the Boolean semiring, obtaining tight bounds for evaluating general \Datalog\ programs over popular semirings (such as absorptive semirings with a total order which are routinely used in data analytics~\cite{RamusatMS21, DynDatalog}) is unclear. \revtwo{In particular, proving correctness of program evaluation is not immediate and the impact of the semiring on the evaluation time remains unknown.}

Endeavors to pinpoint the exact data complexity for \Datalog fragments have focused on the class of Conjunctive Queries (CQs)~\cite{DBLP:journals/mst/JoglekarR18,AGM,PANDA} and union of CQs~\cite{carmeli2021enumeration}, where most have been dedicated to develop faster algorithms. \eat{Different widths have been proposed to characterize the exact runtime of CQ evaluation.} When recursion is involved, however, exact runtimes are known only for restricted classes of \Datalog. Seminal work of Yannakakis~\cite{Yannakakis90} established a $O(n^3)$ runtime for chain \Datalog programs \revthree{(formally defined in Section~\ref{sec:datalogo})}, where $n$ is the size of the active domain. Such programs have a direct correspondence to context-free grammars and capture a fundamental class of static analysis known as context-free reachability (\textsf{CFL}). When the chain \Datalog program corresponds to a regular grammar, the runtime can be further improved to $O(m \cdot n)$\footnote{$m$ denotes the size of the input data and $n$ denotes the size of the active domain for a \Datalogo program throughout the paper.}. An $O(n^3)$ algorithm was proposed for the \Datalog program that captures Andersen's analysis~\cite{Andersen}. Recently, Casel and Schmid~\cite{RPQs} studied the fine-grained complexity of evaluation, enumeration, and counting problems over regular path queries (also a \Datalog fragment) with similar upper bounds. However, none of the above techniques generalize to arbitrary \Datalog\ programs.

\smallskip
\relatedwork{Our Contribution} In this paper, we ask: given a \Datalog\ program $P$ over a naturally-ordered semiring $\bS$, what is the {\em tightest possible runtime as a function of the input size $m$ and the size of the active domain $n$}? Our main contributions are as follows.

\looseness-1 We propose a general, two-phased framework for evaluation of $P$ over $\bS$. The first phase \emph{grounds} $P$ into an equivalent system of polynomial equations $G$. Though constructing groundings na\"ively is rather straightforward, we show that \revthree{groundings of smaller size} are attainable via tree decomposition techniques. The second phase evaluates the least fixpoint of $G$ over the underlying semiring $\bS$. We show that finite-rank semirings and absorptive semirings with total order, two routinely-used classes of semirings in practice, admit efficient algorithms for least fixpoint computation. We apply our framework to prove state-of-the-art and new results for practical \Datalog fragments (e.g. linear \Datalog).

\revone{Further, we establish tightness of our results by demonstrating a matching lower bound on the running time (conditioned on the popular min-weight $\ell$-clique conjecture) and size of the grounded program (unconditional) for a class of \Datalog\ programs}.

\section{Preliminaries}
\label{sec:prelim}


\subsection{\Datalog}\label{sec:datalog}
We review standard \Datalog~\cite{abiteboul1995foundations} that consists of a set of {\em rules} over a set of {\em extensional } and {\em intensional } relations \revthree{(simply referred to as \EDB and \IDB\ respectively, henceforth)}. \EDBs correspond to relations in a given database, each comprising a set of \EDB tuples, whereas \IDBs are derived by the rules. The head of each rule is an \IDB, and the body consists of zero or more \EDBs and \IDBs as a conjunctive query defining how the head \IDB is derived. We illustrate with the example of transitive closure on a binary \EDB $R$ denoting edges in a directed graph, a single \IDB $T$, and two rules:
\begin{align} \label{eq:tc}
T(x_1, x_2) & \obtainedfrom R(x_1, x_2), \; & T(x_1, x_2)  \obtainedfrom T(x_1, x_3) \wedge R(x_3, x_2)
\end{align}
In standard \Datalog, \IDBs are derived given \EDBs (equivalently, evaluation is over the Boolean semiring), whereas as discussed in Section~\ref{sec:intro}, the evaluation can also be considered over other semirings where the \EDBs are annotated with elements of that semiring.  

\subsection{Semirings and Their Properties}\label{sec:prelim-semiring}


\introparagraph{Semirings}  A tuple $\bS = (\boldsymbol{D}, \oplus, \otimes, \mathbf{0}, \mathbf{1})$ is a \textit{semiring} if $\oplus$ and $\otimes$ are binary operators over $\boldsymbol{D}$ for which:
\begin{enumerate}
    \item $(\boldsymbol{D}, \oplus, \zerobf)$ is a commutative monoid with 
    identity $\zerobf$ (i.e., $\oplus$ is associative and commutative, and $a \oplus 0 = a$ for all $a \in D$);
    \item $(\boldsymbol{D}, \otimes, \onebf)$ is a monoid with 
    identity $\onebf$ for $\otimes$;
    \item $\otimes$ distributes over $\oplus$, i.e., $a \otimes (b \oplus c) = (a \otimes b) \oplus (a \otimes c)$ for $a, b, c \in \boldsymbol{D}$; and 
    \item $a \otimes \zerobf = \zerobf$ for all $a \in \boldsymbol{D}$.
\end{enumerate}
$\bS$ is a {\em commutative semiring} if $\otimes$ is commutative.
%
Examples include the Boolean semiring $\mathbb{B}$ = $(\{ \textsf{false}, \textsf{true} \} , \vee, \wedge, \textsf{false}, \textsf{true})$, the natural numbers semiring $(\mathbb{N}, +, \cdot, 0, 1)$, and the real semiring $(\mathbb{R}, +, \cdot, 0, 1)$.

\introparagraph{Naturally-Ordered Semirings} In a semiring $\bS$, the relation $x \sqsubseteq y$ (or $ y \sqsupseteq x$), defined as $\exists z: x \oplus z=y$, is reflexive and transitive, but not necessarily anti-symmetric. If it is anti-symmetric, then $(\boldsymbol{D}, \sqsubseteq)$ is a poset ($\zerobf$ is the minimum element since $\zerobf \oplus a = a$) and we say that $\bS$ is {\em naturally-ordered}. For example, $\mathbb{N}$ is naturally-ordered with the usual order $\leq$. However, $\mathbb{R}$ is not naturally-ordered, since $x \sqsubseteq y$ for every $x, y \in \mathbb{R}$.  We focus on naturally-ordered commutative semirings in this paper (simply referred to as semirings from now).

\introparagraph{$\omega$-complete \& $\omega$-continuous Semirings}
An $\omega$-chain in a poset is a sequence $a_0 \sqsubseteq a_1 \sqsubseteq \dots$ with elements from the poset. A naturally-ordered semiring $\bS$ is {\em $\omega$-complete} if every (infinite) $\omega$-chain $a_0 \sqsubseteq a_1 \sqsubseteq \cdots$
has a least upper bound ${\tt sup}~ a_i$, and is {\em $\omega$-continuous} if also $\forall a \in \boldsymbol{D}$, $a \oplus {\tt sup}~ a_i$ =
${\tt sup}~(a \oplus a_i)$ and $a \otimes {\tt sup}~ a_i = {\tt sup~(a \otimes a_i)}$. 

\introparagraph{Rank} The \textit{rank} of a strictly increasing $\omega$-chain $a_0 \sqsubset a_1 \sqsubset \dots \sqsubset a_r$ is $r$ ($x \sqsubset y$ if $x \sqsubseteq y$ and $x \neq y$). We say that a semiring $\bS$ has rank $r$ if every strictly increasing sequence has rank at most $r$. Any semiring over a finite domain has constant rank (e.g. $\mathbb{B}$ has rank $1$).


We highlight two special classes of naturally-ordered semirings to be studied in this paper: dioids and absorptive semirings.

\introparagraph{Dioids}
A \textit{dioid} $\bS$ is a semiring where $\oplus$ is idempotent, i.e. $a \oplus a = a$ for any $a \in \boldsymbol{D}$. A dioid is naturally ordered, i.e. $a \sqsubseteq b$ if $a \oplus b = b$. 
    The natural order satisfies the monotonicity property: for all $a, b, c \in \boldsymbol{D}, a \sqsubseteq b$ implies $a \oplus c \sqsubseteq b \oplus c$ and $a \otimes c \sqsubseteq b \otimes c$. 
    
\introparagraph{Absorptive Semirings} A semiring $\bS$ is \textit{absorptive} if $\mathbf{1} \oplus a = \mathbf{1}$ for every $a \in \boldsymbol{D}$. An absorptive semiring is also called {\em $0$-stable} in~\cite{POPS}. 
An absorptive semiring is a dioid and $\mathbf{0} \sqsubseteq a \sqsubseteq \mathbf{1}$, for any $a \in \boldsymbol{D}$. An equivalent characterization of absorptiveness is in~\autoref{lemma:absorptive}. Simple examples of absorptive semirings are the Boolean semiring and the \textit{tropical semiring} $\mathsf{Trop}^{+}=\left(\mathbb{R}_{+} \cup\{\infty\}, \min ,+, \infty, 0\right)$ (that can be used to measure the shortest path), where its natural order $x \sqsubseteq y$ is the reverse order $x \geq y$ on $\mathbb{R}_{+} \cup\{\infty\}$.




\subsection{\Datalogo over Semirings} \label{sec:datalogo}
Next we describe the syntax of \Datalogo over a semiring $\bS$. Here \EDBs are considered as $\bS$-relations, i.e., each tuple in each \EDB $R$ is annotated by an element from the domain of $\bS$. Tuples not in the \EDB are annotated by $\zerobf$ implicitly. Standard relations are essentially $\mathbb{B}$-relations. When a \Datalogo program is evaluated on such $\bS$-relations, \IDB $\bS$-relations are derived. Here, conjunction ($\wedge$) is interpreted as $\otimes$, and disjunction as $\oplus$ of $\bS$, as discussed in \cite{green2007provenance}.
\par

\introparagraph{Sum-product Queries} The \Datalog program~\eqref{eq:tc} has two rules of the same head, thus an \IDB tuple can be derived by either rule (i.e. a disjunction over rules). Following Abo Khamis et al. \cite{POPS},  we combine all rules with the same \IDB head into one using disjunction. Hence the program under a semiring $\bS$ is written as:
\begin{align} \label{eq:apsp}
 T(x_1,x_2) &\obtainedfrom R(x_1,x_2) \oplus \left(\bigoplus_{x_3} T(x_1,x_3) \otimes R(x_3,x_2) \right).
\end{align}
where $\oplus$ corresponds to alternative usage (or disjunction) \cite{green2007provenance} and the implicit $\exists x_3$ in the second rule of~\eqref{eq:tc} is made explicit by $\bigoplus_{x_3}$.
\revtwo{Let $\ell$ be the number of variables $x_1, x_2, \cdots, x_{\ell}$  in a Datalog program as in Eqn.~(\ref{eq:apsp}).  For $J \subseteq [\ell]$, where  $[\ell] = \{1, 2, \ldots, \ell\}$,  $\bx_J$ denotes the set of variables $\{ x_j \mid j \in J\}$.} 
In the rest of the paper, w.l.o.g., we consider a \Datalogo program as a set of rules of distinct \IDB heads, where each rule has the following form:
\begin{equation} \label{eq:sum-sumprod-rule}
	 T(\mathbf{x}_H) \obtainedfrom \sumprod_1(\bx_H) \oplus \sumprod_2(\bx_H) \oplus \dots 
\end{equation}
Here $\mathbf{x}_H$ denotes the set of head variables. The body of a rule is a {\em sum} (i.e., $\oplus$) over one or more {\em sum-product queries} (defined below) $\sumprod_1(\bx_H), \sumprod_2(\bx_H), \ldots$ corresponding to different derivations of the \IDB  $T(\mathbf{x}_H)$ over a semiring $\bS$. 
%
%
Formally, a {\em sum-product (in short, \SumProd) query} $\sumprod$ over $\bS$ has the following form:
\begin{equation} \label{eq:sumprod-rule}
     \sumprod(\bx_H): \quad \bigoplus_{\bx_{[\ell] \setminus H}} \revone{\bigotimes_{J \in \mE} T_J(\bx_J)},
\end{equation}
where (1) $\bx_H$ is the set of {head variables}, (2) each $T_J$ is an \EDB or \IDB predicate, and \revone{(3) $([\ell], \mE)$ is the \textit{associated hypergraph} of $\sumprod$ with hyperedges $\mathcal{E} \subseteq 2^{\nodes}$ and vertex set as $\bx_{[\ell]}$}. 

For every \Datalogo program in the sum-product query form, we will assume that there is a unique \IDB that we identify as the \textit{target} (or \textit{output}) of the program. \revone{We use $\arity{P}$ to denote the maximum number of variables contained in any \IDB of a program $P$.}

\introparagraph{Monadic, Linear \& Chain \Datalogo} We say that a \Datalogo program $P$ is {\em monadic} if every \IDB is unary (i.e. $\arity{P} = 1$);  a \Datalogo program is {\em linear} if every \SumProd query in every rule contains at most one \IDB (e.g., the program in Eq.~\eqref{eq:apsp} is linear). A {\em chain} query is a \SumProd query over binary predicates as follows:
$$ \sumprod(x_1, x_{k+1}): \quad \bigoplus_{\bx_{\{2, \ldots, k\}}} T_1(x_1,x_2) \otimes T_2(x_2, x_3) \otimes \dots \otimes T_k(x_k, x_{k+1}).$$
\revthree{A {\em chain} \Datalogo program is a program where every rule consists of one or a sum of multiple chain queries.} 
A chain \Datalogo program corresponds to a Context-Free Grammar (CFG)~\cite{Yannakakis90}. 

\introparagraph{Least Fixpoint Semantics} 
The least fixpoint semantics of a \Datalogo program $P$ over a semiring $\bS$ is defined as the standard \textsf{Datalog}, through its immediate consequence operator (\textsf{ICO}). An \textsf{ICO} applies all rules in $P$ exactly once over a given instance and uses all derived facts (and the instance) for the next iteration. The naive evaluation algorithm iteratively applies \textsf{ICO} until a least fixpoint is reached (if any).
In general, naive evaluation of \Datalogo over a semiring $\bS$ may not converge, depending on $\bS$. Kleene's Theorem~\cite{kleeneTheorem} shows that if $\bS$ is $\omega$-continuous and the semiring $\bS$ is $\omega$-complete, then the naive evaluation converges to the least fixpoint. Following prior work~\cite{NPA, green2007provenance}, we assume that $\bS$ is $\omega$-continuous and $\omega$-complete.

\introparagraph{$\bS$-equivalent} Two \Datalogo programs $P, G$ are $\bS$-equivalent if for any input \EDB instance annotated with elements of a semiring $\bS$, the targets of $P$ and $G$ are identical $\bS$-relations when least fixpoints are reached for both.


\introparagraph{Parameters \& Computational Model} \looseness-1
We use $m$ to denote the \revone{sum of sizes} of the input \EDB $\bS$-relations to a \Datalogo program \revone{(henceforth referred to as the {``\em input size is $m$''})}. We use $n$ to denote the size of the active domain of \EDB $\bS$-relations (i.e., the number of distinct constants that occur in the input). If $\arity{P} = k$, then $n \leq m \cdot k$. We assume {\em data complexity} \cite{Vardi82}, i.e. the program $P$ size (the total number of predicates and the variables) is a constant. We use the standard word-RAM model \revone{with $O(\log m)$-bit words and unit-cost operations}~\cite{AlgBook} for all complexity results. $\widetilde{{O}}$ hides poly-logarithmic factors in the size of the input.
\section{Framework and Main Results}
\label{sec:framework}

This section presents a general framework to analyze the data complexity of \Datalogo\ programs. We show how to use this framework to obtain both state-of-the-art and new algorithmic results.

A common technique to measure the runtime of a \Datalogo\ program is to multiply the number of iterations until the fixpoint is reached with the cost of each \textsf{ICO} evaluation. Although this method can show a polynomial bound in terms of data complexity, it cannot generate the tightest possible upper bound. 
Our proposed algorithmic framework allows us to decouple the semiring-dependent fixpoint computation from the structural properties of the \Datalogo\ program. It splits the computation into two distinct phases, where the first phase concerns the {\em logical structure} of the program, and the second the {\em algebraic structure} of the semiring.

\introparagraph{Grounding Generation} \looseness-1 In the first phase, the \Datalogo\ program is transformed into a $\sigma$-equivalent {\em grounded program}, where each rule contains only constants. A {\em grounding} of a grounded \Datalogo\ program is a system of polynomial equations where we assign to each grounded \EDB\ atom a distinct coefficient $e$ in the semiring and to each grounded \IDB\ atom a distinct variable $x$. Our goal in this phase is to construct a $\sigma$-equivalent grounding $G$ of \revone{the smallest possible size $|G|$. The size of $G$ is measured as the total number of coefficients and variables in the system of polynomial equations}.

\introparagraph{Grounding Evaluation} In the second phase, we need to deal with evaluating the least fixpoint of $G$ over the semiring $\sigma$. The fixpoint computation now depends on the structure of the semiring $\sigma$, and we can completely ignore the underlying logical structure of the program. Here, we need to construct algorithms that are as fast as possible w.r.t. the size of the grounding $|G|$. For example, it is known that over the Boolean semiring, the least fixpoint of $G$ can be computed in time $O(|G|)$~\cite{GK04}.

\begin{table*}[t]
\caption{Summary of the grounding results. The $\widetilde{O}$ notation hides polylog factors in $m$ \revtwo{(total input size)} and $n$ \revtwo{(size of the active domain). $k$ denotes the arity of the program.}  See Section~\ref{sec:grounding} for definitions of $\fhw, \subw$.}
    \centering
    \renewcommand{\arraystretch}{1.5}
    \scalebox{0.99}{
    \begin{tabular}{c | p{2.5cm} | c | l | p{2.5cm}  }\toprule
    {\bf Semiring} &  {\bf Class of Programs} & {\bf \IDB Arity} & {\bf Grounding Size \& Time} & {\bf Result}  \\ \midrule
     all & rulewise-acyclic & $k$ & $O(n^{k-1} \cdot (m+n^k))$ & \autoref{thm:main:acyclic} \\
     all & rulewise free-connex acyclic & $k$ & $O(m+n^k)$ & \autoref{thm:main:acyclic:freeconnex} \\
     all & linear \& rulewise-acyclic & $2$ & $O(n \cdot m )$ & \autoref{prop:linear} \\
     all & fractional hypertree-width $\leq \fhw$ & $k$ & $O(n^{k-1} \cdot (m^\fhw + n^{k \cdot \fhw}))$  & \autoref{prop:main:fhw}  \\ 
      \midrule
      dioid & submodular width $\leq \subw$ & $k$ & $\widetilde{O}(n^{k-1} \cdot (m^\subw + n^{k \cdot \subw}))$  & \autoref{thm:main}  \\ 
      dioid & submodular width $\leq \subw$ & $1$ & $\widetilde{O}(m^\subw)$  & Corollary of~\autoref{thm:main}  \\ 
      dioid & free-connex submodular width $\leq \csubw$ & $k$ & $\widetilde{O}(m^\csubw + n^{k \cdot \csubw})$  & \autoref{thm:main:free-connex}  \\ 
      dioid & linear \& submodular width $\leq \subw$ & $k$ & $\widetilde{O}(n^{k-1} m^{\subw - 1} \cdot (m + n^k))$  & \autoref{prop:main:linear}  \\ 
    \end{tabular}}
    \label{tab:summary}
\end{table*}

\subsection{Grounding Generation}

The first phase of the framework generates a grounded program.

\begin{example}\label{ex:polynomial}
We will use as an example the following variation of program~\eqref{eq:apsp}, over an arbitrary semiring $\bS$:
\begin{align*} 
 T(x_1, x_2) &\obtainedfrom R(x_1, x_2) \oplus \left(\bigoplus_{x_3} T(x_1, x_3) \otimes U(x_3) \otimes \revtwo{R(x_3, x_2)} \right).
\end{align*}
We introduced an unary atom $U$ to capture properties of the nodes in the graph. Consider an instance where $R$ contains two edges, $(a, b)$ and $(b, c)$, and $U$ contains three nodes, $a,b,c$. 
One possible $\sigma$-equivalent grounded \Datalogo program (over any semiring $\sigma$) is:
\begin{align*}
    T(a, b) & \obtainedfrom R(a, b) \oplus T(a, a) \otimes U(a) \otimes R(a, b), &  T(a, a) & \obtainedfrom \mathbf{0}   \\
 T(a, c) & \obtainedfrom T(a, b) \otimes U(b) \otimes R(b, c), & T(b, c)  & \obtainedfrom R(b, c).
\end{align*}
This corresponds to the following system of polynomial equations:
\begin{align*}
    x_{ab} & = e_{ab} \oplus x_{aa} \otimes f_a \otimes e_{ab} & \quad x_{aa} & = \mathbf{0} \\
    x_{ac} & = x_{ab} \otimes f_b \otimes e_{bc} & \quad x_{bc} & = e_{bc}.
\end{align*}
\end{example}
As there is a one-to-one correspondence between a grounded program and its grounding, we will only use the term grounding from now on. 
The naive way to generate a $\sigma$-equivalent grounded program is to take every rule and replace the variables with all possible constants. However, this may create a grounding of very large size. Our key idea is that we can optimize the size of the generated grounding by exploiting the logical structure of the rules.   

\introparagraph{Upper Bounds}
Our first main result (Section~\ref{sec:acyclic}) considers the body of every sum-product query in the \Datalogo program is acyclic \revthree{(formally Definition~\ref{def:acyclic})}; we call such a program {\em rulewise-acyclic}

\begin{theorem}\label{thm:main:acyclic}
Let $P$ be a rulewise-acyclic \Datalogo program over some semiring $\bS$, \revtwo{with input size $m$, active domain size $n$}, and $\arity{P} \leq k$. Then, we can construct a $\bS$-equivalent grounding in time (and has size) $O(n^{k-1} \cdot (m + n^{k}))$.
\end{theorem}    

A direct result of the above theorem is that for {\em monadic} \Datalogo, where $\arity{P}=1$, we obtain a grounding of size $O(m)$, which is essentially optimal. The main technical idea behind Theorem~\ref{thm:main:acyclic} is to construct the join tree corresponding to each rule, and then decompose the rules following the structure of the join tree. 

It turns out that, in analogy to conjunctive query evaluation, we can also get good bounds on the size of the grounding by considering a width measure of tree decompositions of the rules. We show in Section~\ref{sec:grounding} the following theorem, which relates the grounding size to the maximum {\em submodular width} (Section~\ref{sec:grounding}) across all rules.

\begin{theorem}\label{thm:main}
Let $P$ be a \Datalogo program where $\arity{P} \leq k$, $\subw$ is the maximum submodular width across all rules of $P$, and $\bS$ is a dioid and suppose \revtwo{the input size is $m$, and the active domain size is $n$}. Then, we can construct a $\bS$-equivalent grounding in time (and has size) $\widetilde{O}(n^{k-1} \cdot (m^\subw + n^{k \cdot \subw}))$.
\end{theorem}

The above theorem requires the semiring to be idempotent w.r.t. the $\oplus$ operation (i.e. a dioid). For a general semiring, the best we show is that we can replace $\subw$ with the weaker notion of fractional hypertree width (Proposition \ref{prop:main:fhw}). Sections~\ref{sec:acyclic} and~\ref{sec:grounding} show refined constructions that improve the grounding size when the \Datalogo\ program has additional structure (e.g., linear rules); we summarize these results in Table~\ref{tab:summary}.

\introparagraph{Lower Bounds} One natural question is: do Theorem~\ref{thm:main:acyclic} and~\ref{thm:main} attain the best possible grounding bounds? This is unlikely to be true for specific \Datalogo\ fragments (e.g. linear \Datalog has a tighter grounding as shown in Proposition~\ref{prop:main:linear}); however, we can show optimality for a class of programs (see Appendix~\ref{app:lowerbound} for the proof).

\begin{theorem} \label{thm:main_lower_bound}
    Take any integer $k \geq 2$ and any rational number $w \geq 1$ such that $k \cdot w$ is an integer. 
    There \revone{exists} a (non-linear) \Datalogo\ program $P$ over the tropical semiring $\mathsf{Trop}^{+}$ with $\arity{P} \leq k$, such that:
    \begin{enumerate}
        \item $\subw(P) = w$, 
        \item any $\mathsf{Trop}^{+}$-equivalent grounding has size $\Omega(n^{k-1+k w})$,
        \item under the min-weight $\ell$-Clique hypothesis~\cite{DBLP:conf/soda/LincolnWW18}, no algorithm that evaluates $P$ has a runtime of $O(n^{k-1+k w - o(1)})$.
    \end{enumerate}
\end{theorem}

\eat{When every predicate (both \IDB and \EDB) has arity at most $k$, we have $m = O(n^k)$, and thus Theorem~\ref{thm:main} gives an upper bound of $O(n^{k-1 + k \cdot \subw})$. Hence, the lower bound shows that the bound is essentially tight with respect to the parameter $n$.}

\subsection{Grounding Evaluation}

Given a grounding $G$ for a \Datalogo\ program, we now turn to the problem of computing the (least fixpoint) solution for this grounding. In particular, we study {\em  how fast can we evaluate $G$ under different types of semirings as a function of its size.} Ideally, we would like to have an algorithm that computes the fixpoint using $O(|G|)$ semiring operations; however, this may not be possible in general. In Section~\ref{sec:eval}, we will show fast evaluation strategies for two different classes of semirings: $(i)$ semirings of finite rank, and $(ii)$ semirings that are totally ordered and absorptive. Our two main results can be stated as follows.

\begin{theorem} \label{thm:algoRank}
We can evaluate a grounding $G$ over any semiring of rank $r$ using $O(r \cdot |G|)$ semiring operations.
\end{theorem}

\begin{theorem} \label{thm:algoAbsorp}
We can evaluate a grounding $G$ over any absorptive semiring with total order using $O(|G| \log |G|)$ semiring operations.
\end{theorem}

Many semirings of practical interest are captured by the above two classes. The Boolean semiring in particular is a semiring of rank 1. Hence, we obtain as a direct corollary the result in~\cite{GK04}:

\begin{corollary} \label{cor:algoBoolean}
We can evaluate a grounding $G$ over the Boolean semiring in time $O(|G|)$.
\end{corollary}

The \textit{set semiring} $(2^{K}, \cup, \cap, \emptyset, K)$ for a finite set $K$ has rank $|K|$. In fact, all bounded distributive lattices have constant rank. As another example, the \emph{access control} semiring $(\{P, C, S, T, 0\}, \min,$ $\max, 0, P)$~\cite{AccessControl} employs a constant number of security classifications, where $P$ = \textsf{PUBLIC}, $C$ = \textsf{CONFIDENTIAL}, $S$ = \textsf{SECRET}, $T$ = \textsf{TOP-SECRET} and $P \sqsubset C \sqsubset S \sqsubset T \sqsubset 0$ is the total order for levels of clearance. For all these semirings, their grounding can be evaluated in time $O(|G|)$ by Theorem~\ref{thm:algoRank}.  The class of absorptive semiring with total order contains $\mathsf{Trop}^{+}$ that has infinite rank. Hence Theorem~\ref{thm:algoRank} cannot be applied, but by Theorem~\ref{thm:algoAbsorp}, we can evaluate $G$ over $\mathsf{Trop}^{+}$ in time $O(|G| \log |G|)$.

\subsection{Applications}
Finally, we discuss algorithmic implications of our general framework. In particular, we demonstrate that our approach captures as special cases several state-of-the-art algorithms for tasks that can be described in \Datalogo. Table~\ref{tab:runtimes} summarizes some of our results that are straightforward applications of our framework.

To highlight some of our results, Dikjstra's algorithm for single-source shortest path is a special case of applying Theorem~\ref{thm:algoAbsorp} after we compute the grounding of the program. As another example, Yannakakis \cite{Yannakakis90} showed that a binary (i.e. \EDB/\IDB arities are at most 2) rulewise-acyclic programs can be evaluated in time $O(n^3)$. By Theorem~\ref{thm:main:acyclic} (let $k = 2$), an equivalent grounding of size $O(n^3)$ can be constructed in time $O(n^3)$ for such programs. Then by Corollary~\ref{cor:algoBoolean}, we can evaluate the original program in $O(n^3)$ time, thus recovering the result of Yannakakis.

If the binary rulewise-acyclic \textsf{Datalog} program is also linear, it can be a evaluated in $O(m \cdot n)$ time~\cite{Yannakakis90}. We generalize this result to show that the $O(m \cdot n)$ bound holds for any linear rulewise-acyclic \Datalogo program with \IDB arity at most $2$, even if the \EDB relations are of higher arity (see Table~\ref{tab:runtimes}, Proposition~\ref{prop:linear}).

As another corollary of Theorem~\ref{thm:main} and Corollary~\ref{cor:algoBoolean}, monadic \Datalogo can be evaluated in time $\widetilde{O}(m^{\subw})$. This is surprising in our opinion, since $\widetilde{O}(m^{\subw})$ is the best known runtime for Boolean CQs. Hence, this result tells us that the addition of recursion with unary \IDB does not really add to the runtime of evaluation!

\begin{table*}[t]
\caption{Summary of runtime results. The $\widetilde{O}$ notation hides polylog factors in $m$ \revtwo{(total input size)} and $n$ \revtwo{(size of the active domain). $k$ denotes the arity of the program.} See Section~\ref{sec:grounding} for the definition of $\subw$.}
    \centering
    \renewcommand{\arraystretch}{1.5}
    \scalebox{1.0}{
    \begin{tabular}{p{3cm} | c | p{3.75cm} | l }\toprule
    {\bf Task} & {\bf Semiring} & {\bf Program} & {\bf Runtime} \\ \midrule
    \textsf{APSP}~\cite{floyd1962algorithm, warshall1962theorem} & Tropical &  $\!\begin{aligned}[t]
        &T(x_1, x_2) \obtainedfrom R(x_1, x_2) \oplus \\ &\bigoplus_{x_3} T(x_1, x_3) \otimes R(x_3, x_3)
    \end{aligned} $  & $\widetilde{O}(m \cdot n)$ \\
        Single-source Shortest Path (\textsf{SSSP})~\cite{dijkstra2022note} & Tropical &   $\!\begin{aligned}[t]&T(x_1) \obtainedfrom U(x_1) \oplus \\ &\bigoplus_{x_2} T(x_2) \otimes R(x_2, x_1) \end{aligned} $ & $\widetilde{O}(m)$    \\  
      \textsf{CFL} reachability~\cite{Yannakakis90} & Boolean & chain programs & $O(n^3)$ \\
       & Boolean & \IDB arity $\leq 2$ \& rulewise-acyclic  & $O(n^3)$ \\
     Regular Path Queries All-pairs~\cite{barcelo2013querying, RPQs} & Boolean & linear chain programs & $O(m \cdot n)$  \\
     \textsf{CFL-APSP}~\cite{valiant1974general, Yannakakis90}  & Tropical & chain programs & $\widetilde{O}(n^3)$  \\
      Monadic acyclic \Datalogo~\cite{gottlob2002datalog, GK04} & Boolean & monadic rulewise-acyclic & $O(m)$  \\
      Monadic \Datalogo~\cite{gottlob2002datalog, GK04} & Boolean & monadic & $\widetilde{O}(m^{\subw})$  \\
      Linear \Datalogo\ \cite{Hung2023LinearDatalogo} & Boolean & \IDB arity $\leq k$ & $\widetilde{O}(n^{k-1} m^{\subw - 1}  \cdot (m+n^k))$   \\
      \Datalogo\ & Boolean & \IDB arity $\leq k$ & $\widetilde{O}(n^{k-1} \cdot (m^\subw + n^{k \cdot \subw}))$   
    \end{tabular}}
    \label{tab:runtimes}
\end{table*}

\section{Grounding Evaluation}
\label{sec:eval}

This section presents algorithms for evaluating a grounding over \revthree{two types of commonly-used semirings~\cite{RamusatMS21, DynDatalog, POPS}}. First, Section~\ref{subsec:2-canonical} presents a procedure that transforms the grounding to one with a more amenable structure, called a {\em 2-canonical grounding}. Then, we present evaluation algorithms for semirings of rank $r$ (Section~\ref{subsec:rank}), and for absoprtive semirings that are totally ordered (Section~\ref{subsec:absorp}).

\subsection{2-canonical Grounding} \label{subsec:2-canonical}

We say that a grounding $G$ is {\em 2-canonical} if every equation in $G$ is of the form $y = x \oplus z$ or $y = x \otimes z$. As a first step of our evaluation, we first transform the given grounding into a $2$-canonical form using the following ~\autoref{lem:2canonical}. \eat{The proof is deferred to Appendix~\ref{app:2canonical}.}

\begin{lemma} \label{lem:2canonical}
A grounding $G$ can be transformed into a $\bS$-equivalent 2-canonical  grounding of size at most $4 |G|$ in time $O(|G|)$.
\end{lemma}
\begin{proof}
	The construction works in two steps. In the first step, we rewrite all monomials. Consider a monomial $\boldsymbol{m} = x_1 \otimes x_2 \otimes \dots x_n$. If $n = 1$, then simply let $\boldsymbol{m} = x_1 \otimes \onebf$. Otherwise ($n \geq 2$), we rewrite $\boldsymbol{m}$ as follows. We introduce $n-1$ new variables $y_1, \dots, y_{n-1}$, replace $\boldsymbol{m}$ with $y_{n-1}$ and add the following equations in the system:
	\begin{align*}
		y_1 = x_1 \otimes x_2, \quad y_2  = y_1 \otimes x_3, \quad \dots \quad y_{n-1}  = y_{n-2} \otimes x_{n}
	\end{align*}
	This transformation increases the size of the system by at most a factor of two. Indeed, the monomial contributes $n$ to $|G|$, and the new equations contribute $1 + 2(n-2) \leq 2n$. 
	
	After the first step, we are left with equations that are either a product of two elements, or a sum of the form $x = x_1 \oplus x_2 \oplus \dots x_n$ (w.l.o.g., we can assume no equations of the form $x = y$). Here, we apply the same idea as above, replacing multiplication with addition. This transformation will increase the size of the system by another factor of two using the same argument as before. The equivalence of the new grounding follows from the associativity property of both $\oplus,\otimes$.
\end{proof}

\revtwo{
    We demostrate Lemma~\ref{lem:2canonical} with an example. 
    \begin{example}  \label{ex:2canonical}
        Consider the polynomial equation showed up in Section~\ref{sec:framework}: 
        $ x_{ab} = e_{ab} \oplus x_{aa} \otimes f_a \otimes e_{ab} $.
        Its rewriting (via substitutions) following the proof of Lemma~\ref{lem:2canonical} is:
        \begin{align*}
            x_{ab} & = e_{ab} \oplus y_1, \quad y_1 = x_{aa} \otimes y_2, \quad y_2 = f_a \otimes e_{ab}
        \end{align*}
        Here $y_1, y_2$ are new \IDB variables introduced to make the system of polynomial equations 2-canonical. 
    \end{example} 
}

	

\subsection{Finite-rank Semirings} \label{subsec:rank}

We now present a grounding evaluation algorithm (Algorithm~\ref{algoRank}) over a semiring of rank $r$, which is used to prove Theorem~\ref{thm:algoRank}. 

 \begin{algorithm}[h]
 \small
 \SetKwInOut{Input}{Input}
 \SetKwInOut{Output}{Output}
    \Input{grounding $G$} 
    \textbf{transform} $G$ to be 2-canonical via ~\autoref{lem:2canonical}  \\
    \textbf{construct} a hash table $E$, such that for every variable $x$ in $G$, $E[x]$ is the set of all equations that contain $x$ in the right-hand side ; $\;$ \textbf{init} an empty queue $q$  \\
    \lForEach{\textsf{EDB} variables $e$ in $G$} { $h(e) \gets e$ }
    \lForEach{\textsf{IDB} variables $x$ in $G$} { $h(x) \gets \zerobf$ }
     \lForEach{\textsf{EDB} variables $e$ in $G$} { \textbf{insert} $e$ into $q$ }
    \While{$q$ is not empty}{
    \textbf{pop} $x$ from $q$ \label{line:pickvar}    \\
    \ForAll(\tcp*[f]{$\circledast \in \{\oplus,\otimes\}$}){$(y = x \circledast z) \in E[x]$ \label{line:check}}{ 
              \If {$h(y) \neq h(x) \circledast h(z)$ }{
                $h(y) \leftarrow h(x) \circledast h(z)$ \label{line:update} \textbf{insert} $y$ into $q$   \\
                }
    }  }
    \Return $h(\cdot)$ 
    \caption{Grounding evaluation over a rank $r$ {semiring}}\label{algoRank}
\end{algorithm}

The key idea of the algorithm is to compute the least fixpoint in a fine-grained way. Instead of updating all equations in every iteration, we will carefully choose a subset of equations to update their left-hand side variables. In particular, at every step we will pick a new variable (Line~\ref{line:pickvar}) and then only update the equations that contain this variable\footnote{We assume that look-ups, inserts, and deletes on hash tables cost constant time. In practice, hashing can only achieve amortized constant time. Therefore, whenever we claim constant time for hash operations, we mean amortized constant time.} (Line~\ref{line:check}-\ref{line:update}). Because the semiring has rank $r$, the value of each variable cannot be updated more than $r$ times; hence, we are guaranteed that each equation will be visited only $2 \cdot r$ times. Moreover, since the grounding is 2-canonical, updating each variable needs only one $\oplus$ or $\otimes$ operation; the latter property would not be possible if we had not previously transformed the grounding into a 2-canonical one. \eat{The full proof is in Appendix~\ref{app:eval}.}

\subsection{Absorptive Semirings with Total Order} \label{subsec:absorp}

We present a Dijkstra-style algorithm (Algorithm~\ref{algoAbsorp}) for grounding evaluation over an absorptive semiring with $\sqsubseteq$ being a total order 
to prove Theorem~\ref{thm:algoAbsorp}. 
It builds upon prior work~\cite{DynDatalog} by further 
optimizing the original algorithm to achieve the almost-linear runtime. \eat{Formal proof of convergence and runtime is in Appendix~\ref{app:eval}.}


\begin{algorithm}
    \small
    \SetKwInOut{Input}{Input}
    \SetKwInOut{Output}{Output}
       \Input{a grounding $G$}
        \textbf{transform} $G$ to be 2-canonical via ~\autoref{lem:2canonical}  \\
       \textbf{construct} a hash table $E$, such that for every variable $x$ in $G$, $E[x]$ is the set of all equations that contain $x$ in the right-hand side ; $\;$ \textbf{init} $\mathcal{F} \gets \emptyset$ \\

       \lForEach{\textsf{EDB} variables $e$ in $G$} { $h(e) \gets e$ }
       \lForEach{\textsf{IDB} variables $x$ in $G$} { $h(x) \gets \mathbf{0}$ }

       \textbf{init} an empty priority queue $q$ of decreasing values w.r.t. $\sqsupseteq$ \\
       
       \ForAll(\tcp*[f]{$\circledast \in \{\oplus,\otimes\}$}){\textsf{IDB} variables $x = y \circledast z$ in $G$} { 
            $h(x) \gets h(y) \circledast h(z)$ ; \textbf{insert} $x$ into $q$ of value $h(x)$  \\
       }
        
       \While{$q$ is non-empty}{
       \textbf{pop} a variable $x$ of the max value from $q$  \\
        $\mathcal{F} \gets \mathcal{F} \cup \{x\}$ ; \label{line:freeze} \\
       \ForAll(\label{line:noupdate}){$(y = x \circledast z) \in E[x] \wedge y \notin \mathcal{F}$} {
                 \If {$h(y) \neq h(x) \circledast h(z)$}{
                            $h(y) \gets h(x) \circledast h(z)$   \\
                           \textbf{insert} $y$ into $q$ of value $h(y)$  } 
       }  }
       \Return $h(\cdot)$ 
       \caption{Grounding evaluation over an absorptive semiring with total order} \label{algoAbsorp}
   \end{algorithm}

The algorithm follows the same idea as Algorithm~\ref{algoRank} by carefully updating only a subset of equations at every step. However there are two key differences. First, while Algorithm~\ref{algoRank} \revthree{is agnostic} to the order in which the newly updated variables are propagated to the equations (hence the use of a queue), Algorithm~\ref{algoAbsorp} needs to always pick the variable with the current maximum value w.r.t. the total order  $\sqsubseteq$. To achieve this, we need to use a priority queue instead of a queue, which is the reason of the additional logarithmic factor in the runtime. The second difference is that once a variable is updated once, it gets ``frozen'' and never gets updated again (see Lines \ref{line:freeze}-\ref{line:noupdate}). We show in the detailed proof of correctness (Appendix~\ref{app:eval}) that it is safe to do this and still reach the desired fixpoint.

\section{Grounding of Acyclic \Datalogo} 
\label{sec:acyclic}

In this section, we study how to find an efficient grounding (in terms of space usage and time required) of a \Datalogo program over a semiring $\bS$. First, we introduce \textit{rulewise-acyclicity} of a program using the notion of \emph{tree decompositions}. 

%
%

\begin{example}
We ground the \textsf{APSP} program \eqref{eq:apsp}. First, $R(x_1,x_2)$ produces $O(|R|)$ groundings, since it is an EDB. The latter \SumProd query has a $O(|R| \cdot n)$ grounding, since for each tuple $R(x_3,x_2)$, we should consider all active domain of $x_1$. Hence, the equivalent grounding is of size $O(|R| + |R| \cdot n) = O(m \cdot n)$.
\end{example}

\eat{We revisit the definition of tree decompositions of a hypergraph.}

\introparagraph{Tree Decompositions}
Let $\sumprod$ be a \SumProd query \eqref{eq:sumprod-rule} and $(\nodes, \mathcal{E})$ be its associated hypergraph. A {\em tree decomposition} of $\sumprod$ is a tuple $(\mathcal{T}, \chi)$ where $(i)$ $\mathcal{T}$ is an undirected tree, and $(ii)$ $\chi$ is a mapping that assigns to every node $t \in V(\mathcal{T})$ 
a set of variables $\chi(t) \subseteq \nodes$, called the {\em bag} of $t$, such that 
	\begin{enumerate}
		\item  
		every hyperedge $J \in \mathcal{E}$ is a subset of some $\chi(t)$, $t \in V(\mT)$;
		\item (running intersection property)
		for each $i \in \nodes$, the set $\{t \mid i \in \chi(t)\}$ is a non-empty (connected) sub-tree of $\mathcal{T}$.
	\end{enumerate}

	\revtwo{A tree decomposition is a {\em join tree} if $\chi(t) \in \mathcal{E}$ for all $t \in V(\mT)$.} 

\introparagraph{Acyclicity} \label{def:acyclic}
A \SumProd query $\sumprod$ is said to be \textit{acyclic} if its associated hypergraph admits a join tree. The GYO reduction~\cite{yu1979algorithm} is a well-known method to construct a join tree for $\sumprod$. We say that a rule of a \Datalogo program is {\em acyclic} if every \SumProd query in its body is acyclic and a program is {\em rulewise-acyclic} if it has only acyclic rules.

\revthree{We formally prove Theorem~\ref{thm:main:acyclic}, by introducing a grounding algorithm for rulewise-acyclic programs attaining the desired bounds; 
the pseudocode is in Algorithm~\ref{grounded-acyclicprogram} in the appendix.}

\revone{
\begin{proof}[Proof of Theorem~\ref{thm:main:acyclic}]
The $\bS$-equivalent grounding $G$ is constructed rule-wise from $P$. Take any acyclic rule from $P$ with head $T(\bx_H)$. Note that its arity $|H| \leq k$. We ground each \SumProd query $\sumprod(\bx_H)$ in this rule one by one. To ground a single $\sumprod$ (which is of the form~\eqref{eq:sumprod-rule}), we construct a join tree $(\mathcal{T}, \chi)$ of $\sumprod$. Each bag in the join tree will correspond to an atom in the body of $\sumprod$. We root $\mathcal{T}$ from any atom that contains at least one of the variables in $\bx_H$ (i.e. head variables), say $s_0$\eat{(ideally we pick the atom that contains the most head variables, but any choice works for this proof)}. This orients the join tree. For any node $s$, define $(V_s, \mE_s)$, where $V_s \subseteq [\ell], \mE_s \subseteq \mE$, to be the hypergraph constructed by only taking the bags from the subtree of $\mathcal{T}$ rooted at $t$ as hyperedges. Note that $(V_{s_0}, \mE_{s_0}) = ([\ell], \mE)$ is the associated hypergraph of $\sumprod$. Let $H_s \subseteq H$ be the head variables that occur in this subtree (so at root, $H_{s_0} = H$). Further, rename the head atom $T(\bx_H)$ to be $U_{e_{(\bot, s_0)}}(\bx_H)$, with $\bot$ being an imaginary parent node of the root $s$ and $e_{(\bot, s_0)} = H$.

Next, we use a recursive method \textsc{Ground} to ground $\sumprod$. Let $r$ be the parent of $s$ in $\mathcal{T}$. A \textsc{Ground}$(s, U_{e_{(r, s)}})$ call at a node $s$ grounds the \SumProd query
\begin{align} \label{eq:methodGround}
    U_{e_{(r, s)}}(\bx_{e_{(r, s)}}) \obtainedfrom \bigoplus_{\bx_{V_s \setminus e_{(r, s)}}} \bigotimes_{J \in \mathcal{E}_s} T_J(\bx_J).
\end{align} 
Hence, calling \textsc{Ground}$(s_0, U_{e_{(\bot, s_0)}})$ at $s_0$ suffices to ground $\sumprod$. We describe the procedure \textsc{Ground}$(s, U_{e_{(r, s)}})$, which has three steps:

$(i)$ \textit{Refactor.} Define for each child $t$ of $s$ (i.e. $(s, t) \in E(\mT)$, $E(\mT)$ being the directed edges of $\mT$), $e_{(s, t)} = (\chi(s) \cap \chi(t)) \cup H_t$ and 
\begin{align*}
    S & = \chi(s) \cup \bigcup_{t: (s, t) \in E(\mT)} e_{(s, t)}.
\end{align*}

We refactor as follows: $U_{e_{(r, s)}}(\bx_{e_{(r, s)}}) = \bigoplus_{\bx_{V_s \setminus e_{(r, s)}}} \bigotimes_{J \in \mathcal{E}_s} T_J(\bx_J)$
{\small
\begin{align*}
    & \eqone \bigoplus_{\bx_{V_s  \setminus e_{(r, s)}}} T_{\chi(s)}(\bx_{\chi(s)}) \otimes \bigotimes_{t: (s, t) \in E(\mT)}  \bigotimes_{J \in \mathcal{E}_t} T_J(\bx_J) \\
    & \eqtwo \bigoplus_{\bx_{S \setminus e_{(r, s)}} } T_{\chi(s)}(\bx_{\chi(s)}) \otimes \bigotimes_{t: (s, t) \in E(\mT)}  \left( \bigoplus_{\bx_{V_t \setminus e_{(s, t)}}} \bigotimes_{J \in \mathcal{E}_t} T_J(\bx_J) \right) \\
    & \eqthree \bigoplus_{\bx_{S \setminus e_{(r, s)}}} T_{\chi(s)}(\bx_{\chi(s)}) \otimes \bigotimes_{t: (s, t) \in E(\mT)} U_{e_{(s, t)}}(\bx_{e_{(s, t)}})
\end{align*}
}
where (1) regroups the product into the node $s$ and subtrees rooted at each child $t$ of $s$, (2) safely pushes aggregation over every child $t$ (since any $i \in V_t \setminus e_{(s, t)}$ only occurs in the subtree rooted at $t$, otherwise by the running intersection property, $i \in \chi(s) \cap \chi(t) \subseteq e_{(s, t)}$, a contradiction), and (3) replaces every child \SumProd query by introducing a new \IDB $U_{e_{(s, t)}}(\bx_{e_{(s, t)}})$ and a corresponding query:
\begin{align} \label{eq:refactor}
    U_{e_{(s, t)}}(\bx_{e_{(s, t)}}) \obtainedfrom
     \bigoplus_{\bx_{V_t \setminus e_{(s, t)}}} \bigotimes_{J \in \mathcal{E}_t} T_J(\bx_J).
\end{align}

$(ii)$ \textit{Ground.} Instead of grounding the \SumProd query~\eqref{eq:methodGround} as a whole, we ground the $\bS$-equivalent (but refactored) query
\begin{align*}
    U_{e_{(r, s)}}(\bx_{e_{(r, s)}}) \obtainedfrom \bigoplus_{\bx_{S \setminus e_{(r, s)}}} T_{\chi(s)}(\bx_{\chi(s)}) \otimes \bigotimes_{t: (s, t) \in E(\mT)} U_{e_{(s, t)}}(\bx_{e_{(s, t)}}).
\end{align*}
Notice that $S$ is exactly the set of variables appear in the body of the refactored query. We ground this query as follows: for each possible tuple (say $\ba_{\chi(s)}$) in $T_{\chi(s)}$, we add a grounded rule for each tuple (say $\ba_{S \setminus \chi(s)}$) that can be formed over the schema $\bx_{S \setminus \chi(s)}$ using the active domain of each variable. This is done by taking the attribute values for each variable from $\ba_{\chi(s)}$ and $\ba_{S \setminus \chi(s)}$, and substituting it in every predicate (\EDB, \IDB, or the head) of the query.

$(iii)$ \textit{Recurse.} Now every child $t$ of $s$ has an introduced \IDB $U_{e_{(s, t)}}$. For each $t$, we call into \textsc{Ground}$(t, U_{e_{(s, t)}})$ to recursively ground the new rule \eqref{eq:refactor}. If there are no children (i.e. $s$ is a leaf node), the \textsc{Ground} call on $s$ terminates immediately.



\smallskip

We argue the grounding size and time at the grounding step $(ii)$, for any node $s$ when grounding the refactored query. Recall that $r$ is the parent node of $s$. We consider all possible tuples over $\bx_S$ (variables in the body of the refactored query). First, $T_{\chi(s)}$ is of size $O(m)$ if it is an \EDB, or $O(n^k)$ if it is an \IDB. Next, for the number of tuples that can be constructed over $\bx_{S \setminus \chi(s)}$, note that
{\small
$$
S = \chi(s) \cup \bigcup_{t: (s, t) \in E(\mT)} {e_{(s, t)}} = \chi(s) \cup \bigcup_{t: (s, t) \in E(\mT)} H_t \subseteq \chi(s) \cup H_s
$$
}
where the inclusion holds since $H_t \subseteq H_s$ if $t$ is a child of $s$. Observe that if $H_s \subsetneq H$, then $|H_s| \leq k-1$; otherwise (i.e., $H_s = H$), since we have rooted the join tree to a node that contains at least one head variable, at least one of the head variables in $H_s$ must belong to $\chi(s)$ (again by the running intersection property), in which case also $|H_s \setminus \chi(s)| \leq k-1$. Thus, 
$|S \setminus \chi(s)| \leq |(\chi(s) \cup H_s) \setminus \chi(s)| = |H_s \setminus \chi(s)| \leq k-1$. Hence, the groundings inserted for any intermediate rule is of size (and in time) $O(n^{k-1} \cdot (m + n^{k}))$, i.e. the product of the cardinality of $T_{\chi(s)}$ and the number of possible tuples that can be formed over the schema $\bx_{S \setminus \chi(S)}$.
\end{proof}
}

\begin{figure}[t]
\centering
\scalebox{0.99}{
\begin{tikzpicture}[level distance=10em,
every node/.style = {line width=1pt, shape=rectangle, text centered, anchor=center, rounded corners, draw, align=center}, grow=right];

\node (V) {$R_{14}(x_1,x_4)$};
\node (E) [draw=none, above of=V, node distance=2.2cm]{$\bot$};
\node (R) [below left of=V, node distance=3cm] {$R_{24}(x_2, x_4)$};
\node (Ty) [below of=R, node distance=2.2cm] {$T_2(x_2)$};
\node (S) [below right of=V, node distance=3cm] {$R_{34}(x_3, x_4)$};
\node (Tz) [below of=S, node distance=2.2cm] {$T_3(x_3)$};

\draw[->] (E) edge [line width=1pt]  node [draw=none, midway, label=right:{$T(x_1) \obtainedfrom \bigoplus_{x_4}  {R_{14}(x_1, x_4)} \otimes~U_4 \otimes V_4$}] {} (V);
\draw[->] (V) edge [line width=1pt] node [draw=none, midway, label=left:{$U_4(x_4)  \obtainedfrom \bigoplus_{x_2} R_{24}(x_2,x_4) \otimes U_2$}] {} (R);
\draw[->] (R) edge [line width=1pt] node [draw=none, midway, label=left:{$U_2(x_2)  \obtainedfrom T_2(x_2)$}] {} (Ty);
\draw[->] (V) edge [line width=1pt] node [draw=none, midway, label=right:{$V_4(x_4)  \obtainedfrom \bigoplus_{x_3} R_{34}(x_3,x_4) \otimes V_3$}] {} (S); 
\draw[->] (S) edge [line width=1pt] node [draw=none, midway, label=right:{$V_3(x_3)  \obtainedfrom T_3(x_3)$}] {}  (Tz);                                            
                                                        
\end{tikzpicture}}
\caption{A join tree with its corresponding rewriting.}
\vspace{-5mm}
\label{fig:join-tree}
\end{figure}
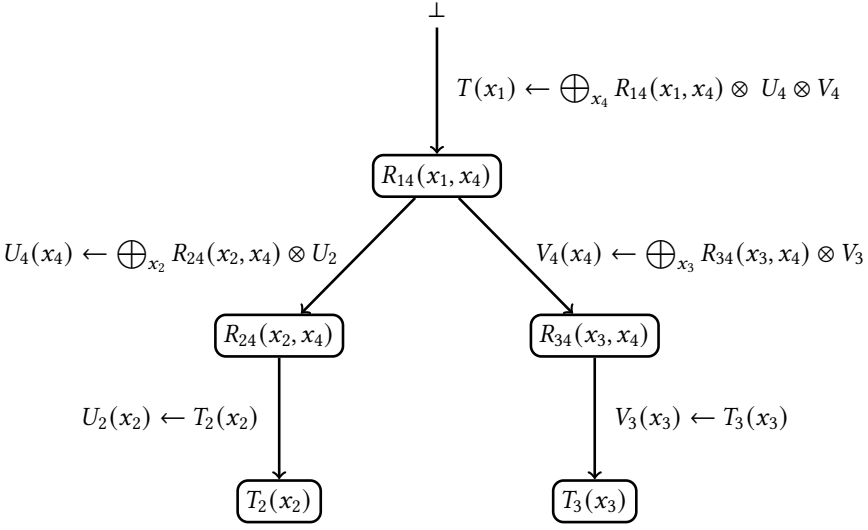

\vspace{-3mm}
\begin{example} \label{ex:acyclic}
We illustrate Theorem~\ref{thm:main:acyclic} with the acyclic rule:
\begin{align*}
T(x_1) \obtainedfrom \bigoplus_{\bx_{234}} T_2(x_2) \otimes T_3(x_3) \otimes R_{24}(\bx_{24}) \otimes R_{34}(\bx_{34}) \otimes R_{14}(\bx_{14}) 
\end{align*}
%
A join tree is drawn in Figure~\ref{fig:join-tree} with new intermediate \IDBs and rules (one per edge of the join tree) introduced recursively by the proof of Theorem~\ref{thm:main:acyclic}. In the transformed program, each rule produces $O(m + n) = O(m)$ groundings. Hence, its overall size is $O(m)$. \revone{A step-by-step walk-through of this example (Example~\ref{ex:acyclic-demo}) and more applications of Theorem~\ref{thm:main:acyclic} are in Appendix~\ref{app:acyclic}}.

%
\end{example}





\introparagraph{Free-connexity} 
Take a tree decomposition $(\mT, \chi)$ of a \SumProd query $\sumprod(\bx_H)$ rooted at a node $r \in V(\mT)$. Let $\textsf{TOP}_r(x)$ be the highest node in $\mT$ containing $x$ in its bag. We say that $(\mT, \chi)$ is \textit{free-connex w.r.t. $r $} if for any $x \in H$ and $y \in [\ell] \setminus H$, $\textsf{TOP}_r(y)$ is not an ancestor of $\textsf{TOP}_r(x)$~\cite{Secure}. We say that $(\mT, \chi)$ is free-connex if it is free-connex w.r.t. some  $r \in V(\mT)$. An acyclic query is free-connex if it admits a free-connex join tree. If every \SumProd query in every rule of a program $P$ admits a free-connex join tree, then $P$ is \textit{rulewise free-connex acyclic}. Theorem~\ref{thm:main:acyclic:freeconnex} (see Table~\ref{tab:summary}) shows that in such a case, there is a much tighter bound (dropping the $n^{k-1}$ term).

\revtwo{
\begin{proposition} \label{thm:main:acyclic:freeconnex}
    Let $P$ be a rulewise free-connex acyclic \Datalogo program over some semiring $\bS$ with input size $m$, active domain size $n$, and $\arity{P}$ is at most $k$. Then, a $\bS$-equivalent grounding can be constructed in time (and has size) $O(m + n^{k})$.
\end{proposition}
}

\introparagraph{Linear Acyclic Programs}
If $\arity{P}\leq 2$, Theorem~\ref{thm:main:acyclic} states that a rulewise-acyclic \Datalogo program admits a grounding of size $O(n^3)$. If the program is also linear (e.g. same generation~\cite{bancilhon1985magic}), we strengthen the upper bound to $O(m \cdot n)$ (see Table~\ref{tab:summary}). The proof and an example (Example~\ref{ex:nonbinary}) can be found in Appendix~\ref{app:acyclic}.

\revtwo{
\begin{proposition} \label{prop:linear}
    Let $P$ be a linear rulewise-acyclic \Datalogo program over some semiring $\bS$ \revtwo{with input size $m$, active domain size  $n$}, and $\arity{P}$ is at most $2$. Then, we can construct a $\bS$-equivalent grounding in time (and has size) $O(m \cdot n)$.
\end{proposition}
}

\section{Grounding of General \Datalogo}
\label{sec:grounding}

This section introduces an algorithm for grounding any \Datalogo program over an arbitary dioid via the \textsf{PANDA} algorithm~\cite{PANDA} (or \textsf{PANDA}, for short). \textsf{PANDA} is introduced to evaluate CQs (i.e. \SumProd queries over the \textit{Boolean semiring}). \eat{Appendix~\ref{sec:PANDA} shows that \textsf{PANDA} can be used for evaluation over \textit{any dioid} (Theorem~\ref{thm:generalizedPANDA}) in time dictated by the submodular width (introduced by Marx~\cite{Marx10, marx2013tractable}) of the query.}

\introparagraph{Submodular width} A function $f : 2^{\nodes} \mapsto \mathbb{R}_+$ is a non-negative {\em set function} on $\nodes$ ($\ell \geq 1$).
The set function is {\em monotone} if $f(X) \leq f(Y)$ whenever $X \subseteq Y$, and is {\em submodular} if $f(X\cup Y)+f(X\cap Y)\leq f(X)+f(Y)$
for all $X,Y\subseteq \nodes$. A non-negative, monotone, submodular set function $h$ such that $h(\emptyset)=0$ is a \textit{polymatroid}. 
Let $\sumprod$ be a \SumProd query \eqref{eq:sumprod-rule}. Let $\Gamma_{\ell}$ be the set of all polymatroids $h$ on $\nodes$ such that $h(J) \leq 1$ for all \revone{$J \in \mE$}. The {\em submodular width} of $\sumprod$ is
\begin{equation} \label{eq:subw}
    \subw(\varphi)  \defeq \; \max_{h \in \Gamma_{\ell}} \min_{(\mathcal{T}, \chi) \in \mathcal{F}} \max_{t \in V(\mathcal{T})} h(\chi(t)),
\end{equation}
where $\mathcal{F}$ is the set of all {\em non-redundant} tree decompositions of $\sumprod$. A tree decomposition is {\em non-redundant} if no bag is a subset of another.
\revone{Abo Khamis et al.~\cite{PANDA} proved that non-redundancy ensures that $\mathcal{F}$ is finite, hence the inner minimum is well-defined.} We define the {\em \revone{free-connex submodular width}} of $\sumprod$, $\csubw(\sumprod)$, by restricting $\mathcal{F}$ to the set of all non-redundant free-connex tree decompositions.

We extend $\subw$ and $\csubw$ to \Datalogo. The {\em submodular width} of a program $P$ (i.e. $\subw(P)$) is the maximum $\subw(\sumprod)$ across all \SumProd queries in $P$. Likewise, the {\em free-connex submodular width} of $P$ (i.e. $\csubw(P)$) is the maximum $\csubw(\sumprod)$ over all \SumProd queries in $P$. If $P$ is rulewise-acyclic, $\subw(P) = 1$.

\revone{This section describes the key ideas of the algorithm underlying Theorem~\ref{thm:main} (formally as Algorithm~\ref{groundedprogram}). We demonstrate the grounding algorithm step-by-step on the following concrete example.}
\revone{
\begin{example} \label{ex:diamond}
    The following \Datalogo computes a diamond-pattern reachability starting from some node a set $U(x_1)$: $T(x_1) \obtainedfrom U(x_1)$,
    \begin{align*} 
        T(x_1) \obtainedfrom \bigoplus_{\bx_{234}} T(x_3) \otimes R_{32}(\bx_{32}) \otimes R_{21}(\bx_{21}) \otimes R_{34}(\bx_{34}) \otimes R_{41}(\bx_{41})
    \end{align*}
    We ground the second rule that involves a 4-cycle join. \eat{A na\"ive strategy uses \textsf{PANDA} to materialize the 4-cycle join and then applies the construction for acyclic programs presented in the proof of Theorem~\ref{thm:main:acyclic}. As demonstrated in Example~\ref{app:ex:diamond}, it yields a grounding of size $\widetilde{O}(m^{3/2} + n^2)$.} In particular, our grounding algorithm constructs a grounding of size $\widetilde{O}(m^{3/2})$.

    \smallskip
    ($i$) \textit{Refactor}. Let $([4], \mE)$ be the associated hypergraph of the above cyclic recursive rule, where $\mE = \{\{3\}, \{3, 2\}, \{2, 1\}, \{3, 4\}, \{4, 1\}\}$. We construct the non-redundant tree decompositions $(\mT_1, \chi_1), (\mT_2, \chi_2)$ (there are only two), both having two nodes $\{v_1, v_2\}$ and one edge $\{(v_1, v_2)\}$, where the bags are
    \begin{align*}
        \chi_1(v_1) = [3], \chi_1(v_2) = \{3, 4, 1\}; \; \chi_2(v_1) = \{4, 1, 2\},  \chi_2(v_2) = \{2, 3, 4\}.
    \end{align*}
    This allows us to rewrite and factorize the cyclic body into two acyclic \SumProd queries, one for each decomposition:
    \begin{align*}
        \sumprod(x_1) & = \bigoplus_{\bx_{234}} T(x_3) \otimes R_{32} \otimes R_{21} \otimes R_{34} \otimes R_{41} \\
        & = \bigoplus_{\bx_{234}} \underbrace{T \otimes R_{32} \otimes R_{21}}_{B_{[3]}(\bx_{[3]})} \otimes \underbrace{R_{34}\otimes R_{41}}_{B_{341}(\bx_{341})}  \oplus \underbrace{R_{21} \otimes R_{41}}_{B_{412}(\bx_{412})} \otimes \underbrace{T \otimes R_{32} \otimes R_{34}}_{B_{234}(\bx_{234})} 
    \end{align*}
    where the second line uses the idempotence of $\oplus$. The underbraces introduce new \IDBs and rules (e.g. $B_{[3]} \obtainedfrom T \otimes R_{32} \otimes R_{21}$ corresponds to the bag $\chi_1(v_1)$ of $\mT_1$ and likewise for $B_{341}$, $B_{234}$ and $B_{412}$). 

    \smallskip
    ($ii$) \textit{Grounding Bags}. A na\"ive grounding of (say) $B_{[3]}$ is of size $O(m \cdot n)$ (cartesian product of $R_{32}$ with domain of $x_1$), which is suboptimal. Instead, we use \textsf{PANDA} to ground the new \IDBs. 
    
    
    Let $Q(\bx_{[4]})$ be the set of tuples such that a tuple $\ba_{[4]} \in Q$ if and only if its annotation $T(\ba_3) \otimes R_{32}(\ba_{32}) \otimes R_{21}(\ba_{21}) \otimes R_{34}(\ba_{34}) \otimes R_{41}(\ba_{41}) \neq \mathbf{0}$. Here, $\ba_{J}$ with $J \subseteq [4]$ is a shorthand for the projection of $\ba_{[4]}$ onto the variables in $J$. The \textsf{PANDA} algorithm outputs one table for every new \IDB (to differentiate, denote the \textsf{PANDA} output tables as $B^*_{[3]}$, $B^*_{341}$, $B^*_{234}$ and $B^*_{412}$) such that: (1) each table is of size $\widetilde{O}(m^{3/2})$ since $\subw(P) = 3/2$; (2) $Q(\bx_{[4]})$ is equal to the union of two CQs, one for each decomposition, i.e.
    \begin{align*}
        ( B^*_{[3]}(\bx_{[3]}) \bowtie B^*_{341}(\bx_{341}) ) \cup ( B^*_{412}(\bx_{412}) \bowtie B^*_{234}(\bx_{234}) ) = Q(\bx_{[4]})
    \end{align*}

    
    Using \textsf{PANDA} output tables, we ground the four new \IDBs, e.g. for each tuple $\ba_{[3]} \in B^*_{[3]}$, we add a grounded rule: $B_{[3]}(\ba_{[3]}) \obtainedfrom T(\ba_3)\otimes R_{32}(\ba_{32}) \otimes R_{21}(\ba_{21})$. The same step is applied to all \IDBs in total time and size $\widetilde{O}(m^{3/2})$. This ensures that each tuple in $Q(\bx_{[4]})$ is present in at least one of $B_{[3]} \otimes B_{341}$ or $B_{412} \otimes B_{234}$, with the correct annotation being preserved \eat{(occurring in both is also correct due to idempotence)}, i.e. for any $\ba_{[4]} \in Q(\bx_{[4]})$, 
    \begin{align*}
        T \otimes R_{32} \otimes R_{21} \otimes R_{34} \otimes R_{41} = ( B_{[3]} \otimes B_{341} ) \oplus \left( B_{412}\otimes B_{234} \right).
    \end{align*}

    \smallskip
    ($iii$) \textit{Grounding Acyclic Sub-queries}. With the grounded bags, we now rewrite $\sumprod(x_1)$ as a sum of two acyclic \SumProd queries:
    \begin{align*}
        \sumprod(x_1) =  \left[\bigoplus_{\bx_{234}} B_{[3]}(\bx_{[3]}) \otimes B_{341}(\bx_{341})\right] \oplus \left[ \bigoplus_{\bx_{234}} B_{412}\otimes B_{234} \right]
    \end{align*}

    We ground the two acyclic sub-queries (one per decomposition) using the construction in the proof of Theorem~\ref{thm:main:acyclic}. For example, if we root both trees at $v_1$, we get a rewriting:
    {\small
    \begin{align*}
        U_{31}(\bx_{31}) & \obtainedfrom \bigoplus_{x_4}  B_{341} \qquad \qquad U_{24}(\bx_{24}) \obtainedfrom \bigoplus_{x_3}  B_{234}  \\
        \varphi(x_1) & \obtainedfrom \left[\bigoplus_{\bx_{23}} B_{[3]} \otimes U_{31}(\bx_{31}) \right] \oplus \left[\bigoplus_{\bx_{24}} B_{412} \otimes U_{24}(\bx_{24}) \right]
    \end{align*}
    }
    guaranteeing that every intermediate \IDB has a grounding of size $\widetilde{O}(m^{3/2})$. Thus, the total size of the grounding for $P$ is $\widetilde{O}(m^{3/2})$.
\end{example}}

We show a free-connex version (Proposition~\ref{thm:main:free-connex}) by restricting the first refactoring step to use only free-connex decompositions. Though $\csubw \geq \subw$, a $n^{k-1}$ factor is shaved off from the bound.

\revtwo{
\begin{proposition}\label{thm:main:free-connex}
    Let $P$ be a \Datalogo program \revtwo{with input size $m$, active domain size  $n$}, $\arity{P} \leq k$. $\csubw$ is its free-connex submodular width. Let $\bS$ be a dioid. Then, a $\bS$-equivalent grounding can be constructed in time (and has size) $\widetilde{O}(m^\csubw + n^{k \cdot \csubw})$.
\end{proposition}
}

\introparagraph{Linear Programs} For linear programs, a careful analysis yields the following improved result (proof is in Appendix~\ref{app:miscellaneous}).

\revtwo{
\begin{proposition} \label{prop:main:linear}
    Let $P$ be a linear \Datalogo program where\\ $\arity{P} \leq k$, input size $m$, and active domain size  $n$. Let $\bS$ be a dioid. Then, a $\bS$-equivalent grounding can be constructed in time (and has size) $\widetilde{O}(n^{k-1} m^{\subw - 1} \cdot (m + n^k))$.
\end{proposition}
}

\introparagraph{Fractional Hypertree-width} 
So far, all our results only apply to dioids. However, we can extend our results to \emph{any} semiring by using the \textsf{InsideOut} algorithm ~\cite{faq}. Similar to $\subw(\sumprod)$ \eqref{eq:subw}, the {\em fractional hypertree-width} of a \SumProd query $\sumprod$, i.e. $\fhw(\sumprod)$, is defined as
\begin{align*}
    \fhw(\sumprod) \defeq \; \min_{(\mathcal{T}, \chi) \in \mathcal{F}}  \max_{h \in \Gamma_{\ell}} \max_{t \in V(\mathcal{T})} h(\chi(t)) 
\end{align*}
The fractional hypertree-width of a program $P$ is the maximum $\fhw$ over all \SumProd queries, denoted as $\fhw(P)$. We show ~\autoref{prop:main:fhw} for the grounding size over \emph{any} semiring (see Table~\ref{tab:summary}).
\vspace{-3mm}
\begin{proposition}\label{prop:main:fhw}
\revtwo{
    Let $P$ be a \Datalogo program \revtwo{with input size $m$, active domain size  $n$}, and $\arity{P} \leq k$. Let $\bS$ be a naturally-ordered semiring. Then, a $\bS$-equivalent grounding can be constructed in time (and has size) $O(n^{k-1} \cdot (m^{\fhw} + n^{k \cdot \fhw}))$.
}
\end{proposition}

\section{Related Work}

\relatedwork{Complexity of \Datalog} Data complexity of special fragments of \Datalog has been explicitly studied. Gottlob et al.~\cite{GK04} showed that monadic acyclic \Datalog can be evaluated in linear time \revone{w.r.t the size of the program plus the size of the input data}. Gottlob et al.~\cite{gottlob2002datalog} defined the fragment \emph{\Datalog LITE} as the set of all stratified \Datalog queries whose rules are either monadic or guarded. The authors showed that this fragment can also be evaluated in linear time as monadic acyclic \Datalog. This result also follows from a generalization of Courcelle's theorem~\cite{courcelle1990graph} by Flum et al.~\cite{flum2002query}. Our framework subsumes these results as an application. \eat{More recently,~\cite{RPQs} studied the fine-grained static and dynamic complexity of Regular Path Queries (\textsf{RPQ}) evaluation, enumeration, and counting problems. The main difference from our paper is that we study the data complexity for a fixed \textsf{RPQ}.} Lutz et al.~\cite{lutz2022efficiently} studied efficient enumeration of \emph{ontology-mediated queries} that are acyclic and free-connex acyclic. The reader may refer to Green at al.~\cite{GreenHLZ13} for an in-depth survey of \Datalog rewriting, \Datalog with disjunctions, and with integrity constraints. However, there is no principled study on general \Datalog programs in parameterized complexity despite its prevalence in the literature of join query evaluation~\cite{skew14, faq, PANDA, Marx10}.

\smallskip
\relatedwork{\Datalog over Semirings} \eat{~\cite{POPS} proposed \textsf{Datalog}$^\circ$, an extension of \Datalog evaluated over a \POPS (short for Partially Ordered Pre-Semirings). \POPS subsumes semirings, and in addition, it allows for an imposed partial order other than the natural order. ~\cite{POPS} characterizes the convergence of \textsf{Datalog}$^\circ$ through grounding and the stability of the underlying \POPS. Coarse convergence rate (i.e. the number of Kleene iterations to converge) are also shown. }Besides \textsf{Datalog}$^\circ$~\cite{POPS}, recent work~\cite{Hung2023LinearDatalogo} has looked at the convergence rate for linear \textsf{Datalog}$^\circ$. The evaluation of \Datalog over absorptive semirings with a total order has also been studied~\cite{DynDatalog} where the key idea is to transform the program (and its input) into a weighted hypergraph and use Knuth's algorithm~\cite{knuth1977generalization} for evaluation. Our paper recovers and extends this result by formalizing the proof of correctness via the concept of asynchronous Kleene chains (Appendix~\ref{app:async}), and showing a precise runtime for evaluating a grounding over such semirings. None of the works mentioned above focus on finding the smallest possible grounding, a key ingredient to show the tightest possible bounds.

\smallskip
\relatedwork{\Datalog Provenance Computation and Circuits} \looseness-1 Algorithms for \Datalog provenance computation provides an alternative way to think of \Datalog evaluation. Deutch et al.~\cite{deutch2014circuits} initiated the study of circuits for database provenance. They show that for a \Datalog program having $|G|$ groundings for \IDBs, a circuit for representing \Datalog provenance for the \textsf{Sorp}$(X)$ semiring (absorptive semirings are a special case of \textsf{Sorp}) can be built using only $|G|+1$ layers. As an example, for \textsf{APSP}, the circuit construction and evaluation cost $O(n^4)$ time. An improvement of the result~\cite{Jukna15} showed that for \textsf{APSP}, a monotone arithmetic circuit of size $O(n^3)$ can be constructed by mimicking the dynamic programming nature of the Bellman-Ford algorithm. We use this improvement to show a circuit unconditional lower bound on the grounding size (see Appendix~\ref{app:lowerbound}).
\section{Conclusion}

This paper introduces a general two-phased framework that uses the structure of a \Datalog\ program to construct a tight grounding, and then evaluates it using the algebraic properties of the semiring. Our framework successfully recovers state-of-the-art results for popular programs (e.g. chain programs, \textsf{APSP}), and uncovers new results (e.g. for linear \Datalogo). We also show a matching lower bound (both for running time and space requirement) for a class of Datalog programs. Future work includes efficient evaluation over broader classes of semirings~\cite{POPS}, circuit complexity of \Datalog over semirings and general grounding lower bounds.

\begin{acks}
This work was done in part while Zhao, Koutris, Roy, and Tannen were visiting the Simons Institute for the Theory of Computing. We gratefully acknowledge the support of NSF via awards IIS-2008107, IIS-2147061, and IIS-1910014.
\end{acks}

\bibliography{refs}

\appendix
\newpage
\section{Missing Details from Section~\ref{sec:prelim}} \label{app:prelim}
\begin{lemma}[\cite{DynDatalog}] \label{lemma:absorptive}
    Let $\bS$ be a naturally-ordered dioid. Then $\bS$ is absorptive iff $ a \otimes b \sqsubseteq a$ for all $a, b \in \boldsymbol{D}$.
\end{lemma}
\begin{proof}[Proof of Lemma \ref{lemma:absorptive}]
    ($\Rightarrow$) Since $\bS$ is absorptive, $a \oplus a \otimes b = a \otimes (\onebf \oplus b) = a \otimes \onebf  = a$. Thus, $a \otimes b \sqsubseteq a$; \\
    ($\Leftarrow$) For any $a \in \bS$, $a = \onebf \otimes a \sqsubseteq \onebf$ (i.e. $\onebf \oplus a = \onebf$)
\end{proof}

\section{Missing Details from Section~\ref{sec:eval}} \label{app:eval}

\subsection{Asynchronous Fixpoint Evaluation} \label{app:async}

The result in this part on asynchronous evaluation for \Datalog is folklore, but we formalize it in a rigorous and general way. The definitions and lemmas established are used in the proof of Theorem~\ref{thm:algoRank} and Theorem~\ref{thm:algoAbsorp}, to show convergence.

Let $G$ be a grounding containing $k$ rules over $\bS$. Suppose we have $k$ posets $\bL_1, \bL_2, \dots, \bL_{k}$ with minimum elements $\bot_1, \dots, \bot_k$ respectively. Let $\{f_i\}_{i \in [k]}$ be a family of $k$ motonone functions, where
$$ f_i: \bL_1 \times \cdots \times \bL_k \rightarrow \bL_i, $$
associates to the right-hand side of the $i$-th rule (the multivariate polynomial) in $G$. Let $\bot = (\bot_1, \dots, \bot_k)$ and the \textsf{ICO} of $G$ is 
$$
\boldsymbol{f} = (f_1, \dots, f_k): \bL_1 \times \cdots \times \bL_k \rightarrow \bL_1 \times \cdots \times \bL_k.
$$
The Kleene sequence (or ascending Kleene chain) attempts to evaluate the grounding (if there is a fixpoint) through the following (ascending) sequence: 
$$ \boldsymbol{h}^{(0)} = \bot, \quad \boldsymbol{h}^{(1)} = \boldsymbol{f}(\boldsymbol{h}^{(0)}), \quad \boldsymbol{h}^{(2)} = \boldsymbol{f}(\boldsymbol{h}^{(1)}), \quad \dots $$
Here $\boldsymbol{h}^{(t)} = (h_1^{(t)}, \dots, h_k^{(t)})$, $t \geq 0$, is the vector of values for the posets at the $t$-th iteration (applying $\boldsymbol{f}$ for $t$ times).

We say that the Kleene sequence converges in $t$ steps if $\boldsymbol{h}^{(t+1)} = \boldsymbol{h}^{(t)}$. Note that the Kleene sequence updates $\boldsymbol{h}^{(t)}$ by using the values of the previous iteration simultaneously. We show that we can obtain convergence if $\boldsymbol{h}^{(t)}$ is updated asynchronously.

\begin{definition}
An \textit{asynchronous Kleene chain} is a chain where $\boldsymbol{\hbar}^{(0)} = \bot$ and for any $t \geq 0$:
$$\hbar^{(t+1)}_i =  f_i(\hbar_1^{(t_1)}, \hbar_2^{(t_2)}, \dots,  \hbar_n^{(t_n)}), \text{ for some } 0 \leq t_i \leq t. $$
\end{definition}

In other words, the current values of each variable can be updated by looking at possibly stale values of the other variables. When $t_1 = t_2 = \dots = t$, we get exactly a Kleene chain.
The following proposition tells us that any asynchronous Kleene chain is always below the Kleene chain.

\begin{proposition}\label{prop:asyncKleene}
Let $\boldsymbol{\hbar}^{(0)}, \boldsymbol{\hbar}^{(1)}, \ldots$ be an asynchronous Kleene chain. For every $t \geq 0$,  $\boldsymbol{\hbar}^{(t)} \sqsubseteq \boldsymbol{h}^{(t)}$.
\end{proposition}

\begin{proof}
We will use the inductive hypothesis that for every $t' \leq t$, $\boldsymbol{\hbar}^{(t')} \sqsubseteq \boldsymbol{h}^{(t')}$. For the base case $t=0$, we have $\boldsymbol{\hbar}^{(0)} = \bot = \boldsymbol{h}^{(0)}$. For the inductive case, assume that $\boldsymbol{\hbar}^{(t')} \sqsubseteq \boldsymbol{h}^{(t')}$ for every $t' \leq t$. Then, we can write:
\begin{align*}
\hbar^{(t+1)}_j  & =  f_i(\hbar_1^{(t_1)}, \hbar_2^{(t_2)}, \dots,  \hbar_n^{(t_n)}) \\
 & \sqsubseteq f_i(h^{(t_1)}, h^{(t_2)}, \dots,  h^{(t_n)}) \\
 & \sqsubseteq f_i(h^{(t)}, h^{(t)}, \dots,  h^{(t)}) \\
 & = h^{(t+1)}_i
\end{align*}
The first inequality follows from the inductive hypothesis, and the second inequality is implied by the monotonicity of the Kleene chain. Both inequalities use the monotonicity of the function $f_i$. 
\end{proof}

\subsection{Proof of Theorem~\ref{thm:algoRank}}

In this part, we complete the proof of the theorem, by showing that Algorithm~\ref{algoRank} correctly computes the solution with the desired runtime.

\begin{proof}[Proof of \autoref{thm:algoRank}]
    Algorithm~\ref{algoRank} terminates when no variable changes value, in which case we have obtained a fixpoint.
    By \autoref{prop:asyncKleene}, we have reached the least fixpoint (since we can view the updates as an asynchronous Kleene chain). 
    
    We now analyze its runtime. The hash table construction and the variable initialization cost time $O(|G|)$. For the while loop, observe that each operation in Line~\ref{line:pickvar}-\ref{line:update} costs $O(1)$ time. Hence, it suffices to bound the number of times we visit each equation in $G$. For this, note that to consider the equation $y = x \circledast z$, either $x$ or $z$ needs to be updated. However, since the semiring has rank $r$, each variable can be updated at most $r$ times. Hence, an equation can be considered at most $2r$ times in the while loop.
\end{proof}

\subsection{Proof of Theorem~\ref{thm:algoAbsorp}}

In the following, we assume $G$ to be a 2-canonical grounding after applying ~\autoref{lem:2canonical}. First, we establish the following two lemmas throughout the while-loop of Algorithm~\ref{algoAbsorp}. The first one is the following:

\begin{lemma} \label{lem:inv1}
     For any \textsf{IDB} variable $x$ in $G$,
     $\mathbf{0} \sqsubseteq x^{(1)} \sqsubseteq x^{(2)} \sqsubseteq \cdots$, where $x^{(0)} = \mathbf{0}, x^{(1)},  x^{(2)}, \ldots$ are the values of $x$ in the asynchronous Kleene chain constructed from the while-loop of Algorithm~\ref{algoAbsorp}.
\end{lemma}

\begin{proof}
    The base case is the first pop, say it pops the variable $x_1$. Trivially, $\mathbf{0} \sqsubseteq  x_1^{(1)} = x_1^{(2)} = \cdots$ because $x_1$ is fixed onwards. Then, for an \textsf{IDB} variable $y \notin E[x]$, $y^{(1)} = y^{(2)}$, else $y^{(1)} = \mathbf{0} \circledast z^{(0)} \sqsubseteq  x_1^{(1)} \circledast z^{(1)} = y^{(2)}$.
    
    Now in the inductive case (at the $(i+1)$-th pop of the variable $x_{i}$), the inductive hypothesis is that for each \textsf{IDB} variable $x$, $\mathbf{0}  \sqsubseteq x^{(1)} \sqsubseteq  x^{(2)} \sqsubseteq \cdots \sqsubseteq  x^{(i)}$. Then, for any \textsf{IDB} variable $y$, $y^{(i)} = x_i^{(i-1)} \circledast z^{(i - 1)} \sqsubseteq x_i^{(i)} \circledast z^{(i)} = y^{(i+1)} $ by monotonicity of both $\oplus, \otimes$.
\end{proof}

The second invariant on the while-loop of Algorithm~\ref{algoAbsorp} is stated as the following:
\begin{lemma} \label{lem:inv2}
    The values of \textsf{IDB} variables popped off the queue are non-increasing w.r.t. to $\sqsubseteq$.
\end{lemma}

\begin{proof} Let $x^{(i)}$ ($i \geq 1$) denote the value of variable $x$ at the $i$-th pop. If $(x_1, x_2, \dots)$ is the sequence of variables that get popped off the queue, then we want to show that $\mathbf{1} \sqsupseteq x_1^{(1)} \sqsupseteq x_2^{(2)} \sqsupseteq  \cdots \sqsupseteq  \mathbf{0}$. We prove it via induction. Observe that $x_i$ will bind to $x_i^{(i)}$ after the $i$-th pop. Thus, $x_i^{(i)} = x_i^{(j)}$ for any $j \geq i$.

The base case is trivial, since $\mathbf{1} \sqsupseteq s \sqsupseteq \mathbf{0}$ holds for any $s \in \boldsymbol{D}$. Now, inductively, suppose we pop $x_{i}$ ($i \geq 2$) at some point and  $\mathbf{1} \sqsupseteq x_1^{(1)} \sqsupseteq x_2^{(2)} \sqsupseteq \cdots \sqsupseteq x_i^{(i)}$  is the sequence of already popped values. We will show that $ x_{i}^{(i)} \sqsupseteq  x_{i+1}^{(i+1)}$.

We distinguish two cases. Suppose that $x_{i}$ is not in the right-hand side of the equation of $x_{i+1}$. Then,  $x_{i}^{(i)} \sqsupseteq x_{i+1}^{(i)} =  x_{i+1}^{(i+1)} $, where the first inequality comes from the fact that $x_{i}^{(i)}$ is popped off as a max element in the queue at iteration $i$.

Otherwise, we have in $G$ exactly one equation of the form $x_{i+1} = x_i \circledast z$ ($z$ may be a \textsf{EDB} coefficient). Now there are two cases depending on the operator:
\begin{enumerate}
    \item $x_{i+1} = x_i \otimes z$. By superiority of $\otimes$, $x_{i}^{(i)} \sqsupseteq x_{i}^{(i)} \otimes z^{(i)} = x_{i+1}^{(i+1)}$. 
    \item $x_{i+1} = x_i \oplus z$. If $x_{i}^{(i)} \sqsupseteq z^{(i)}$, then $x_{i+1}^{(i+1)} = x_{i}^{(i)}$. Otherwise, $x_{i}^{(i)} \sqsubset z^{(i)}$. If $z$ is an \textsf{IDB} variable, we get a contradiction since $x_i$ is not a max value at the $i$-th pop. If $z$ is an \textsf{EDB} coefficient, then $z^{(i-1)} = z^{(i)} $ and thus
    \begin{align*}
        x_{i+1}^{(i)} = x_i^{(i-1)} \oplus z^{(i-1)} = x_i^{(i-1)} \oplus z^{(i)} \sqsupseteq  z^{(i)} \sqsupset x_{i}^{(i)} 
    \end{align*}
    However, $ x_{i+1}^{(i)} \sqsupset x_{i}^{(i)} $ is a contradiction since then $x_{i+1}$ would have been the $i$-th pop instead of $x_i$.
\end{enumerate}
\end{proof}

\revthree{Now we are ready to present the proof for \autoref{thm:algoAbsorp}. In particular, we will use the two lemmas established above (Lemmas~\ref{lem:inv1} and~\ref{lem:inv2}) to show that the algorithm terminates and returns the desired fixpoint. Then we will analyze its runtime, as predicated in \autoref{thm:algoAbsorp}.}

\begin{proof}[Proof of Theorem~\ref{thm:algoAbsorp}]
    We show that Algorithm~\ref{algoAbsorp} terminates and yields a fixpoint, and by \autoref{prop:asyncKleene}, we get the least fixpoint for free because the Dijkstra-style while-loop is a special case of asynchronous Kleene chain. First, it is easy to see that Algorithm~\ref{algoAbsorp} terminates because the size of $q$ always decreases by $1$ (popped \textsf{IDB} variables will not be pushed back into the queue, using $\mathcal{F}$ for bookkeeping of already popped variables). 
    
    We now show that the returned $h(\cdot)$ is a fixpoint. To prove the fixpoint, we show that once the queue pops off an \textsf{IDB} variable $x$ as some value $h(x)$, the equation with $x$ at the left-hand side in $G$, say $x = y \circledast z$, always holds from that point onwards. We prove this by contradiction. Suppose that $x = y \circledast z$ causes the earliest violation in the underlying asynchronous Kleene chain. By \autoref{lem:inv1}, the values of $y, z$ are non-decreasing w.r.t. $\sqsubseteq$. Therefore, the only way to violate $x = y \circledast z$ is that the value of an \textsf{IDB} variable $y$ (or $z$) strictly increases after $x$ is popped off the queue. On the other hand, the queue will eventually pop off that \textsf{IDB} variable ($y$ or $z$) $\sqsupset$ $x$, since Algorithm~\ref{algoAbsorp} always terminates and $y, z$ can not grow smaller by \autoref{lem:inv1}. Yet, it is a direct contradiction to \autoref{lem:inv2}. 

    Lastly, we provide the runtime analysis. Using similar arguments in the proof of \autoref{thm:algoRank}, Lines 1-7 of Algorithm~\ref{algoAbsorp} can execute in $O(|G|)$ time. For its while-loop, each iteration binds an \textsf{IDB} variable $x$ and calls the insert operations of the priority queue $q$ at most $\mathsf{deg}(x)$ times, where $\mathsf{deg}(x)$ denotes the number of the equations in $E[x]$. A 2-canonical grounding has that $\sum_{x \in E} \mathsf{deg}(x) = 2 |G|$. Using a classic max heap implementation of a priority queue, inserts (or updates) of values can run in $O(\log |G|)$ time. Thus, the overall runtime is 
    $$
    O\left(|G| + \sum_{x \in E}\mathsf{deg}(x) \cdot \log |G| \right) = O(|G| \cdot \log |G|).
    $$
\end{proof}

\section{Missing Details from Section~\ref{sec:acyclic}} \label{app:acyclic}

In this part, we include:

\begin{enumerate}
    \item[(\ref{app:acyclic:algo})] \revone{Algorithm~\ref{grounded-acyclicprogram} for grounding a rulewise-acyclic \Datalogo programs and a step-by-step walk-through on Example~\ref{ex:acyclic}.}
    \item[(\ref{app:acyclic:applications})] Two addtional acyclic \Datalogo applications (Example~\ref{sec:andersen} and Example~\ref{sec:anbncn}) of Theorem~\ref{thm:main:acyclic}.
    \item[(\ref{app:acyclic:freeconnex})] Missing details for rulewise free-connex acyclic \Datalogo programs.
    \item[(\ref{app:acyclic:linearAtMost2})] Missing details for linear rulewise-acyclic \Datalogo programs of arity at most 2.
\end{enumerate}

\subsection{Algorithm for Grounding Rulewise-Acyclic \Datalogo Programs} \label{app:acyclic:algo}

The formal algorithm for grounding a rulewise-acyclic \Datalogo program is given in Algorithm~\ref{grounded-acyclicprogram}. We provide a step-by-step walk-through on Example~\ref{ex:acyclic} to complement Figure~\ref{fig:join-tree}. 

\revtwo{

\begin{example} \label{ex:acyclic-demo}
    Recall the acyclic rule presented in Example~\ref{ex:acyclic}:
    \begin{align*}
        T(x_1) \obtainedfrom \bigoplus_{\bx_{234}} T_2(x_2) \otimes T_3(x_3) \otimes R_{24}(\bx_{24}) \otimes R_{34}(\bx_{34}) \otimes R_{14}(\bx_{14}) 
    \end{align*}
\end{example}
We now show how the proof of Theorem~\ref{thm:main:acyclic} (or equivalently, the \textsc{Ground}$()$ method of Algorithm~\ref{grounded-acyclicprogram}) constructs the intermediate rules drawn in Figure~\ref{fig:join-tree}. For succinctness, we call the nodes in the join tree by their corresponding atoms. Following the rooted join tree in Figure~\ref{fig:join-tree}, we start from calling \textsc{Ground} at the root node $R_{14}$, i.e. \textsc{Ground}$(R_{14}, T(x_1))$:
\begin{enumerate}
    \item [$(i)$] (\textit{Refactor}) The root node $R_{14}$ has two children and thus we introduce two new \IDBs $U_4(x_4)$ for the node $R_{24}$ and $V_4(x_4)$ for the node $R_{34}$. This corresponding to the following \textit{refactoring} of the original rule:
        \begin{align*}
            T(x_1) &=  \bigoplus_{\bx_{234}} T_2(x_2) \otimes T_3(x_3) \otimes R_{24}(x_2, x_4) \otimes R_{34}(x_3, x_4) \otimes R_{14}(x_1, x_4)  \\
            &= \underbrace{\bigoplus_{x_4} R_{14}(x_1, x_4) \otimes \left( \bigoplus_{x_2} R_{24}(x_2, x_4) \otimes T_2(x_2) \right) \otimes \left( \bigoplus_{x_3} R_{34}(x_3, x_4) \otimes T_3(x_3) \right)}_{\textit{(push $\bigoplus$ over two children)}} \\
            &= \underbrace{\bigoplus_{x_4} R_{14}(x_1, x_4) \otimes U_4(x_4) \otimes V_4(x_4)}_{\textit{(introduce two new \IDBs)}}
        \end{align*}
        along with two new rules for the intermediate \IDBs:
        \begin{align*}
            U_4(x_4) &\obtainedfrom \bigoplus_{x_2} R_{24}(x_2, x_4) \otimes T_2(x_2) \\
            V_4(x_4) &\obtainedfrom \bigoplus_{x_3} R_{34}(x_3, x_4) \otimes T_3(x_3).
        \end{align*}

    \item [$(ii)$] (\textit{Ground}) We ground the refactored rule $\underline{T(x_1) \obtainedfrom \bigoplus_{x_4}  {R_{14}(x_1, x_4)} \otimes U_4(x_4) \otimes V_4(x_4)}$ (i.e. the top rule shown in Figure~\ref{fig:join-tree}) by taking all tuples in the \EDB $R_{14}$ and that naturally subsumes all possible tuples with non-zero annotations (i.e. ${R_{14}(x_1, x_4)} \otimes U_4(x_4) \otimes V_4(x_4) \neq \zerobf$). Recall that tuples not in the \EDB $R_{24}$ has an implicit annotation of $\zerobf$ and thus do not contribute to the sum-product. This gives a grounding of size $O(|R_{14}|) = O(m)$.
    \item [$(iii)$] (\textit{Recurse}) We recursively call \textsc{Ground} on the children nodes, i.e. \textsc{Ground}$(R_{24}, U_4(x_4))$ and \textsc{Ground}$(R_{34}, V_4(x_4))$. To avoid repeatence, we only demonstrate \textsc{Ground}$(R_{24}, U_4(x_4))$ next.
\end{enumerate}

Again, following the join tree rooted at $R_{24}$, we call \textsc{Ground}$(R_{24}, U_4(x_4))$:
\begin{enumerate}
    \item [$(i)$] (\textit{Refactor}) The node $R_{24}$ has one child $T_2(x_2)$ and thus we introduce a new \IDB $U_2(x_2)$. This corresponds to the following \textit{refactoring} of the original rule:
        \begin{align*}
            U_4(x_4) &= \bigoplus_{x_2} R_{24}(x_2, x_4) \otimes T_2(x_2) & \textit{(no more variables to push $\bigoplus$)} \\
            &= \bigoplus_{x_2} R_{24}(x_2, x_4) \otimes  U_2(x_2) & \textit{(introduce a new \IDB)} 
        \end{align*}
        along with a new rule for the intermediate \IDB:
        \begin{align*}
            U_2(x_2) &\obtainedfrom T_2(x_2).
        \end{align*}
    \item [$(ii)$] (\textit{Ground}) We ground the refactored rule $\underline{U_4(x_4) \obtainedfrom \bigoplus_{x_2} R_{24}(x_2, x_4) \otimes  U_2(x_2)}$ by taking all tuples in the \EDB $R_{24}(x_2, x_4)$. This gives a grounding of size $O(|R_{24}|) = O(m)$.
    \item [$(iii)$] (\textit{Recurse}) We recursively call \textsc{Ground}$(T_2, U_2(x_2))$ next.
\end{enumerate}

Finally, we call \textsc{Ground}$(T_2, U_2(x_2))$. However, since $T_2$ is a leaf node, there is no refactoring and further recursion. We simply ground the rule: $\underline{U_2(x_2) \obtainedfrom T_2(x_2)}$ by taking all values in the active domain of $x_2$. This gives a grounding of size $O(n)$ (recall $n \leq m$).
}

\begin{algorithm}[!t]
	\SetKwInOut{Input}{Input}
	\SetKwInOut{Output}{Output}
	\SetKwFunction{ground}{\textsc{Ground}}
	\SetKwProg{myproc}{Method}{}{}
	\Input{rulewise-acyclic \Datalogo Program $P$ over a semiring $\bS$}
	\Output{$\bS$-equivalent Grounding $G$}
        \textbf{let} $G = \emptyset$; \\
        \ForEach{rule $T(\bx_H) \obtainedfrom \sumprod_1(\bx_H) \oplus \sumprod_2(\bx_H) \ldots \in P$} {
            \ForEach{\SumProd query $\sumprod(\bx_H) \in T(\bx_H)$}{ 

                \textbf{construct} a join tree $(\mT, \chi)$ of $\sumprod(\bx_H)$ rooted at a node $s$, where its corresponding atom $T_{\chi(s)}$ contains the most head variables; \\
                \textbf{let} $H_t \subseteq H$ be the set of head variables that occur in the subtree of $\mT$ rooted at $t \in V(\mT)$; \\
                \textbf{let} $T(\bx_H) = U_{e_{(\emptyset, s)}}(\bx_{e_{(\emptyset, s)}})$; \tcp*[f]{$e_{(\emptyset, s)} = H_s = H$}\\
                \ground{$s$, $ U_{e_{(\emptyset, s)}}(\bx_{e_{(\emptyset, s)}})$}; \\
            }
        }

	\myproc(\tcp*[f]{$r$ is the parent of $s$}){\ground{$s$, $U_{e_{(r, s)}}(\bx_{e_{(r, s)}})$}}{ 
        \textcolor{gray}{// \textit{$(i)$ Refactor the subtree rooted at $s$ by introducing one new \IDB per child node $t$}} \\ 
        
        \ForEach(\tcp*[f]{$t$ is a child of $s$}){$(s, t) \in E(\mT)$}{
            \textbf{introduce} a new \IDB $U_{e_{(s, t)}}(\bx_{e_{(s, t)}})$ where $e_{(s, t)} = (\chi(s) \cap \chi(t)) \cup H_t$; \\
        }

        \textcolor{gray}{// \textit{$(ii)$ Ground the passed-in \IDB} $U_{e_{(r, s)}}(\bx_{e_{(r, s)}}) \obtainedfrom \bigoplus_{\bx_{S \setminus e_{(r, s)}}} T_{\chi(s)} (\bx_{\chi(s)})\otimes \bigotimes_{t: (s, t) \in E(\mT)} U_{e_{(s, t)}}(\bx_{e_{(s, t)}})$} \label{line:groundingStart}  \\
        \textbf{let} $S = \chi(s) \cup \bigcup_{t: (s, t) \in E(\mT)} e_{(s, t)}$; \\
        \ForEach(\tcp*[f]{$O(m + n^k)$}){$\ba_{\chi(s)} \in T_{\chi(s)}(\bx_{\chi(s)}) $}{
            \ForEach(\tcp*[f]{$O(n^{k-1})$}){$\ba_{S \setminus \chi(s)} \in \mathsf{Dom}(\bx_{S \setminus \chi(s)}) $}{
                \textbf{let} $\ba_{S} = (\ba_{\chi(s)}, \ba_{S \setminus \chi(s)})$;  \tcp*[f]{concat tuple} \\
                \If(\tcp*[f]{$\Pi_{e_{(s, t)}}\ba_{S} $ is the projection of $\ba_{S}$ to $\bx_{e_{(s, t)}}$} ){$U_{e_{(r, s)}}(\Pi_{e_{(r, s)}} \ba_{S} ) \notin G$ }{
                    \textbf{insert} into $G$ a new grounded rule
                    $U_{e_{(r, s)}}(\Pi_{e_{(r, s)}} \ba_{S} ) \obtainedfrom T_{\chi(s)} ( \ba_{\chi(s)}) \otimes \bigotimes_{t: (s, t) \in E(\mT)} U_{e_{(s, t)}}(\Pi_{e_{(s, t)}} \ba_{S} )$; \\
                }
                \Else{
                    \textbf{add} to the body of $U_{e_{(r, s)}}(\Pi_{e_{(r, s)}} \ba_{\vars{\phi_s}} )$ by $T_{\chi(s)} ( \ba_{\chi(s)}) \otimes \bigotimes_{t: (s, t) \in E(\mT)} U_{e_{(s, t)}}(\Pi_{e_{(s, t)}} \ba_{S} )$; \label{line:groundingEnd} \\
                }
            }
        }
        \textcolor{gray}{// \textit{$(iii)$ Recursively call \textsc{Ground} on each child node $t$ (pass in the new \IDB $U_{e_{(s, t)}}$)}} \\
        \ForEach{$(s, t) \in E(\mT)$}{
            \ground{$t$, $U_{e_{(s, t)}}(\bx_{e_{(s, t)}})$}; \\
            }
    }
	\caption{Ground a rulewise-acyclic \Datalogo Program over a semiring $\bS$} \label{grounded-acyclicprogram}
\end{algorithm}

\subsection{Two Additional Applications of Theorem~\ref{thm:main:acyclic}} \label{app:acyclic:applications}
We show two applications (Example~\ref{sec:andersen} and Example~\ref{sec:anbncn}) to illustrate the grounding technique for acyclic \Datalogo programs in Theorem~\ref{thm:main:acyclic}. 

\begin{example} \label{sec:andersen}
    Andersen's analysis can be captured by the following acyclic \textsf{Datalog} program of arity 2:
    \begin{align*}
     T(x_1,x_2) &\obtainedfrom \textsf{AddressOf}(x_1,x_2) \vee \left(\bigvee_{x_3} T(x_1,x_3) \wedge \textsf{Assign}(x_3,x_2) \right) \\ &\vee \left(\bigvee_{\bx_{34}}  T(x_1,x_4) \wedge T(x_4,x_3) \wedge \textsf{Load}(x_3,x_2) \right)  \vee \left(\bigvee_{\bx_{34}}  T(x_1, x_4) \wedge \textsf{Store}(x_4, x_3) \wedge T(x_2,x_3) \right).
    \end{align*}
    By Theorem~\ref{thm:main:acyclic}, we can obtain an equivalent grounded \Datalogo program of size $O(n^3)$ in time $O(n^3)$ since for binary \EDB predicates, $O(m + n^2) = O(n^2)$.
\end{example}

Next, we ground an acyclic \Datalogo Program for the language $a^p b^p c^p$, where $p \geq 1$. Note that this language is not context-free.

\begin{example} \label{sec:anbncn}
    \revtwo{Suppose we have a directed labeled graph, where each directed edge is labeled as $a$, $b$, or $c$ (possibly more than one labels for a single edge). 
	We construct a \Datalogo program over the natural number semiring $\mathbb{N} = (\mathbb{N}, +, \cdot, 0, 1)$ to compute: for every pair of node $(x_a, y_c)$ in this graph, the number of directed labeled paths from $x_a$ to $y_c$ that satisfies the expression $a^p b^p c^p$ for $p \geq 1$. The target in the following program (of arity 6) is the \IDB $Q$ and the \EDB predicates are $a, b, c$, where for each tuple $(x_a, y_a)$ in the \EDB $a$, its natural number annotation $a(x_a, y_a)= 1$ (since it means that there is a directed edge from node $x_a$ to node $y_a$ labeled as $a$ in the graph), and similarly for \EDBs $b$ and $c$. Tuples absent from the \EDB predicates are assumed to have annotation $0$ (i.e. no such labeled edge in the graph).
    The program is as follows:
	\begin{align*}
	T(x_a, y_a, x_b, y_b, x_c, y_c) & \obtainedfrom  a(x_a, y_a) \cdot b(x_b, y_b) \cdot c(x_c, y_c) \;  + \\ &\left( \sum_{y_a', y_b', y_c'} T(x_a, y_a', x_b, y_b', x_c, y_c') \cdot a(y_a', y_a)  \cdot b(y_b', y_b) \cdot c(y_c', y_c) \right).  \\
    T_{\textsf{path}}(x_a, y_a, y_b, y_c) & \obtainedfrom T(x_a, y_a, x_b, y_b, x_c, y_c), \; y_a = x_b, \; y_b = x_c. \\
	Q(x_a,y_c) & \obtainedfrom \sum_{y_a,y_b} T_{\textsf{path}}(x_a, y_a, y_b, y_c).
	\end{align*}
    We will use our grounding method to bound the cost of evaluating this \Datalogo program. The third rule obviously has a grounding of size $n^4$. For the second rule that involves selection predicates, one can rewrite it as the acyclic rule
    $$
    T_{\textsf{path}}(x_a, y_a, y_b, y_c) \obtainedfrom T(x_a, y_a, x_b, y_b, x_c, y_c), I(y_a, x_b), I(y_b, x_c).
    $$
    where $I(y, x)$ are new \EDB predicates that represent the identity function over the active domain, i.e. $I(y, x) = 1$ if $y = x$ and $0$ otherwise. Hence, $I(y_a, x_b)$ and $I(y_b, x_c)$ are of size $O(n)$ and this rule has a grounding of size $O(n^2 \cdot n \cdot n) = O(n^4)$.

    Now we ground the first rule (also acyclic) as follows: the number of groundings of the first \SumProd query $a(x_a, y_a) \cdot b(x_b, y_b) \cdot c(x_c, y_c)$ is $m \cdot n^4$. For the second \SumProd query, we can decompose it as follows:
	\begin{align*}
	T_1(x_a, y_a, x_b, y_b', x_c, y_c') & \obtainedfrom \sum_{y_a'} T(x_a, y_a', x_b, y_b', x_c, y_c') \cdot a(y_a', y_a).  \\
	T_2(x_a, y_a, x_b, y_b, x_c, y_c') & \obtainedfrom \sum_{y_b'} T_1(x_a, y_a, x_b, y_b', x_c, y_c') \cdot b(y_b', y_b).  \\
	T(x_a, y_a, x_b, y_b, x_c, y_c) & \obtainedfrom \sum_{y_c'} T_2(x_a, y_a, x_b, y_b, x_c, y_c') \cdot c(y_c', y_c).  
	\end{align*}
	It is easy to see that the groundings for each rule are bounded by $m \cdot n^5$. Hence, the overall grounding is of size $O(n^4 + m \cdot n^4 + m \cdot n^5) = O(m \cdot n^5)$.}
\end{example}

\subsection{Rulewise free-connex acyclic \Datalogo programs}  \label{app:acyclic:freeconnex}

\begin{proof}[Proof of~\autoref{thm:main:acyclic:freeconnex}]
    The construction is similar to that of Algorithm~\ref{grounded-acyclicprogram}, except that we root the join tree at the node $r$, where the tree is free-connex w.r.t. $r$. Now, we give a more fine-grained argument on the grounding size at Line \ref{line:groundingStart}-\ref{line:groundingEnd} of Algorithm~\ref{grounded-acyclicprogram} when $\mathsf{Ground}(s, U_{e_{(r, s)}}(\bx_{e_{(r, s)}}))$ is called. Recall that we want to ground the intermediate rule
    $$ U_{e_{(r, s)}}(\bx_{e_{(r, s)}}) \obtainedfrom \bigoplus_{\bx_{S \setminus H}} T_{\chi(s)} (\bx_{\chi(s)})\otimes \bigotimes_{t: (s, t) \in E(\mT)} U_{e_{(s, t)}}(\bx_{e_{(s, t)}})
    $$
 where $e_{(r, s)} = (\chi(r) \cap \chi(s)) \cup H_s$. Follow the proof of Theorem~\ref{thm:main:acyclic}, it suffices to cover $\chi(s) \cup H_s$ and for that, we split into the following two cases:
\begin{itemize}
    \item If $s \in V(\mT)$ contains only head variables (i.e. $\chi(s) \subseteq H$). Then, $\chi(s) \cup H_s \subseteq H$ and taking all instantiations in the active domain of $\chi(s) \cup H_s$ suffices to ground the body for the above rule. There are $O(n^k)$ such instantiations since $|H| \leq k$.
    \item Otherwise, $s$ contains at least one non-head variable (say $i$). Then, $H_s \subseteq \chi(s)$ since if there is a head variable $j \notin \chi(s)$, but it occurs in the subtree rooted at $s$, then $\mathsf{TOP}_r(i)$ must be an ancestor of $\mathsf{TOP}_r(j)$, which contradicts the fact that the join tree is free-connex w.r.t. $r$. Thus, $\chi(s) \cup H_s = \chi(s)$ and we can trivially cover $\chi(s)$ by $O(m + n^k)$ groundings.
\end{itemize}
For both cases, the size of the groundings inserted into $G$ for the intermediate rule is $O(m + n^k)$.
\end{proof}

\subsection{Linear rulewise-acyclic \Datalogo programs of arity at most 2} \label{app:acyclic:linearAtMost2}

\begin{proof}[Proof of~\autoref{prop:linear}]
    We illustrate the grounding technique for a single acyclic \SumProd $\sumprod(x, y)$ containing at most one \IDB. If $\sumprod(x, y)$ has no \IDB (so essentially a \SumProd query), then there is a grounding of size $O(n\cdot m)$ that can be constructed in time $O(n \cdot m)$, since now every atom has a grounding of $O(m)$ instead of $O(m + n^2)$ in the proof of Theorem~\ref{thm:main}. Thus, we only analyze the case where $\sumprod(x, y)$ contains exactly one binary \IDB $T$. We also assume that no \EDB atom contains both variables of $T$ (otherwise, we can ground the \IDB $T$ by $O(m)$ and reduce back to the \SumProd query case).

    Similar to Algorithm~\ref{grounded-acyclicprogram}, we construct a join tree $(\mT, \chi)$ rooted at a node $r$ containing the head variable $x$. If now the \IDB $T$ is a leaf node, then we simply follow Algorithm~\ref{grounded-acyclicprogram} to get a grounding of size $O(m \cdot n + n^2) = O(m \cdot n)$. Otherwise, let $\mT_t$ be the subtree rooted at $t$, where $\chi(t)$ corresponds to the \IDB atom $T$, and $H_t$ be the set of head variables that occur in $\mT_t$. If $y \in \chi(t)$, or if $y \notin H_t$, then we subsitute $\mT_t$ in $\sumprod(x, y)$ by a new \IDB $T'(\bx_{\chi(t)})$ and ground a new free-connex acyclic rule
    \begin{equation} \label{eq:proof:substitute}
        T'(\bx_{\chi(t)}) \obtainedfrom \bigoplus_{\bx_{[\ell] \setminus \chi(t)}} \bigotimes_{s \in V(\mT_t)} T_{\chi(s)}(\bx_{\chi(s)}).
    \end{equation}
    By Theorem~\ref{thm:main:acyclic:freeconnex}, we can construct for this rule a grounding of size $O(m + n^2)$ in time $O(m + n^2)$. Then, we reduce back to the previous case when we ground $\sumprod(x, y)$ since the only \IDB $T'$ becomes a leaf node. 

    The trickier case is when $y \in H_t \setminus \chi(t)$. First, we identify the joining variable between $T$ and its parent atom. Observe that there can not be more than one join variables because we assumed that no atom contains $\chi(t)$. There are two cases: $x \in \chi(T)$, then $x$ must be the join variable since $x \in \chi(r)$, i.e., the \IDB is $T(x, w)$; or $x \notin \chi(T)$ (hence $x \notin H_t$). Let the \IDB be $T(z, w)$, and w.l.o.g., let $z$ be the only join variable with its parent \EDB. 
    
    We only analyze the $T(z, w)$ case. The $T(x, w)$ case is identical. We find the node $t_y = \mathsf{TOP}_r(y)$ in $\mT_t$ and the path from $t$ to $t_y$. First, by the exact substitution as \eqref{eq:proof:substitute}, we can construct a grounding of size $O(m)$ in time $O(m)$ for $\mT_{t_y}$ (i.e., the subtree rooted at $t_y$) and replace it with a new \EDB $R'_{\chi(t_y)}(\bx_{\chi(t_y)})$. Then, for every edge $(p, q) \in E(\mT_t)$ along the path from $t$ to $t_y$, we substitute the \EDB $R_{\chi(p)}(\bx_{\chi(p)})$ corresponds to $p$ by a new \EDB $R'_{\chi(p)}(\bx_{\chi(p)})$ via the following free-connex acyclic \SumProd query
    \begin{align*}
        R'_{\chi(p)}(\bx_{\chi(p)}) & \obtainedfrom \bigoplus_{\bx_{[\ell] \setminus \chi(p)}} R_{\chi(p)}(\bx_{\chi(p)}) \otimes \bigotimes_{(p, q') \in E(\mT_t): q' \neq q} \bigotimes_{s \in V(\mT_{q'})} R_{\chi(s)}(\bx_{\chi(s)}).
    \end{align*}
    where the body has no \IDB. Thus, the rule has a grounding of size $O(m)$ obtained in $O(m)$ time. Such substitutions prune the subtree $\mT_t$ into a path from $t$ to $t_y$. Thus we introduce another \IDB $T'(z, y)$ and the linear rule whose body contains only atoms on the path (i.e. $T(z, w) \otimes \cdots \otimes R_{\chi(t_y)}(\bx_{\chi(t_y)})$). Using a join-project style rewritting, it is easy to see that  the rule for $T'(z, y)$ has a grounding of size $O(m \cdot n)$ that can be obtained in time $O(m \cdot n)$. Finally, we collapse the subtree $\mT_t$ again into a single atom $T'(z, y)$ at the leaf. Following Algorithm~\ref{grounded-acyclicprogram}, we get a grounding of size $O(m \cdot n + n^2) = O(m \cdot n)$.
\end{proof}

Next, we demonstrate the application of~\autoref{prop:linear} and how it leads to an improvement over directly applying Algorithm~\ref{grounded-acyclicprogram}.

\begin{example}  \label{ex:nonbinary}
	Consider the following linear rule (essentially a chain query) of a \textsf{Datalog} program with the only \IDB $T(x_2, x_3)$ in the body:
	\begin{align*}
		T(x_1, x_5) &\obtainedfrom \bigoplus_{\bx_{234}} R_{12}(x_1, x_2) \otimes T(x_2, x_3) \otimes R_{34}(x_3, x_4) \otimes R_{45}(x_4, x_5).
	\end{align*}
Suppose that we elect $R_{12}$ as the root bag for the join tree shown in~\autoref{fig:nonbinary}.
\begin{figure}[t]
	\centering
	\scalebox{0.95}{
		\begin{tikzpicture}[level distance=10em,
			every node/.style = {line width=1pt, shape=rectangle, text centered, anchor=center, rounded corners, draw, align=center}, grow=right];
			
			\node (V) {$R_{12}(x_1, x_2)$};
			\node (E) [draw=none, above of=V, node distance=1.2cm]{};
			\node (R) [below of=V, node distance=2cm] {$T(x_2,x_3)$};
			\node (T) [below of=R, node distance=2cm] {$R_{34}(x_3, x_4)$};			
			\node (U) [below of=T, node distance=2cm] {$R_{45}(x_4,x_5)$};			
			\draw[->] (E) edge [line width=1pt]  node [draw=none, midway, label=right:{$T(x_1, x_5) \obtainedfrom \bigoplus_{x_2}  R_{12}(x_1, x_2) \otimes T_{25}(x_2,x_5)$}] {} (V);
			\draw[->] (V) edge [line width=1pt] node [draw=none, midway, label=left:{$T_{25}(x_2, x_5)  \obtainedfrom \bigoplus_{x_3} T(x_2,x_3) \otimes T_{35}(x_3,x_5)$}] {} (R);
			\draw[->] (R) edge [line width=1pt] node [draw=none, midway, label=left:{$T_{35}(x_3,x_5)  \obtainedfrom \bigoplus_{x_4} R_{34}(x_3, x_4) \otimes T_{45}(x_4,x_5) $}] {} (T);
			\draw[->] (T) edge [line width=1pt] node [draw=none, midway, label=left:{$T_{45}(x_4,x_5)  \obtainedfrom R_{45}(x_4,x_5)$}] {} (U);
			
	\end{tikzpicture}}
	\caption{A join tree with its corresponding rewriting for~\autoref{ex:nonbinary} when using Algorithm~\ref{grounded-acyclicprogram}.}
\end{figure}
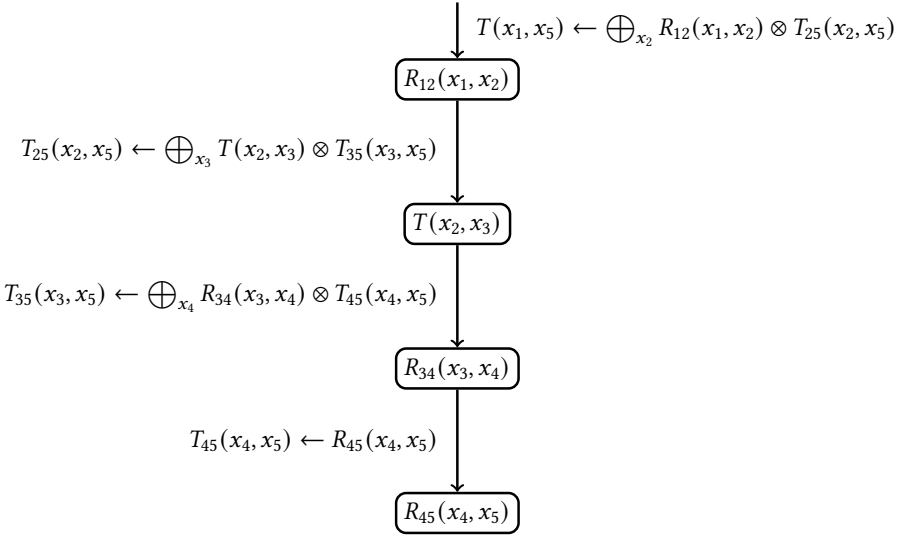 \label{fig:nonbinary}
Algorithm~\ref{grounded-acyclicprogram} running the join tree is suboptimal: the grounding of the rule $T_{25}(x_2, x_5)  \obtainedfrom \bigoplus_{x_3} T(x_2,x_3) \otimes T_{35}(x_3,x_5)$ requires $O(n^3)$ since every possible constant combination of $(x_2, x_3, x_5)$ in the active domain must be considered from $T(x_2,x_3) \otimes T_{35}(x_3,x_5)$.

In constrast, \autoref{prop:linear} avoids the situation where an intermediate rule contains more than one \IDB in the body. The algorithm runs in a top-down fashion first from the \IDB node and constructs the grounding for $T_{25}(x_2,x_5)$ via the following:
\begin{align*}
    T_{24}(x_2,x_4) & \obtainedfrom \bigoplus_{x_3} T(x_2,x_3) \otimes R_{34}(x_3,x_4) \\
    T_{25}(x_2,x_5) & \obtainedfrom \bigoplus_{x_4} T_{24}(x_2,x_4) \otimes R_{45}(x_4,x_5).
\end{align*}
Now both intermediate rules have a grounding of $O(m \cdot n)$.  Finally, we ground the rule $T(x_1, x_5) \obtainedfrom \bigoplus_{x_2} R_{12}(x_1, x_2) \otimes T_{25}(x_2,x_5)$ in $O(m \cdot n)$ time and size. Indeed, one can check that it is an $\bS$-equivalent rewriting of the program since
\begin{align*}
    \bigoplus_{x_2} R_{12}(x_1, x_2) \otimes T_{25}(x_2,x_5) & = \underbrace{\bigoplus_{x_2} R_{12}(x_1, x_2) \otimes \left( \bigoplus_{x_4} T_{24}(x_2,x_4) \otimes R_{45}(x_4,x_5) \right)}_{{\textit{(substitute $T_{25}(x_2,x_5)$)}}}   \\
    & = \underbrace{\bigoplus_{\bx_{24}} R_{12}(x_1, x_2) \otimes T_{24}(x_2,x_4) \otimes R_{45}(x_4,x_5)}_{\textit{(pull out the summation $\bigoplus_{x_4}$)}}  \\
    & = \underbrace{\bigoplus_{\bx_{24}} R_{12}(x_1, x_2) \otimes \left( \bigoplus_{x_3} T(x_2,x_3) \otimes R_{34}(x_3,x_4) \right) \otimes R_{45}(x_4,x_5)}_{\textit{(substitute $T_{24}(x_2,x_4)$)}}  \\
    & = \underbrace{\bigoplus_{\bx_{234}} R_{12}(x_1, x_2) \otimes T(x_2,x_3) \otimes R_{34}(x_3,x_4) \otimes R_{45}(x_4,x_5)}_{\textit{(pull out the summation $\bigoplus_{x_3}$)}}  \\
    & = \underbrace{T(x_1, x_5)}_{\textit{(by definition)}}  \\
\end{align*}
\end{example}

\section{Lower Bounds} \label{app:lowerbound}

In this section, we will prove the lower bounds stated in Theorem~\ref{thm:main_lower_bound}. Our lower bounds are based on the following proposition, which constructs a class of \Datalogo\ programs with the appropriate width that decides whether a undirected graph $G$ contains a clique.

\begin{proposition} \label{prop:clique:program}
Take an integer $k \geq 2$ and a rational number $w \geq 1$ such that $k \cdot w$ is an integer. 
There is a (non-recursive) \Datalogo program $P$ with arity $k$ such that $\subw(P) \leq w$, and the target Boolean \IDB predicate $Q()$ for $P$ decides whether a undirected graph $G$ has a clique of size $\ell = (k-1)+k \cdot w$.
\end{proposition}

\begin{proof}
 Let $R(x,y)$ be the binary predicate that represents the edges of $G$, and $V(x)$ be the unary predicate that represents the nodes of $G$. It will be helpful to distinguish the $\ell$ variables of the clique we want to compute into $k$ variables $x_1, \dots, x_k$ and the remaining $\ell-k$ variables $y_1, y_2, \dots, y_{\ell-k}$. Note that $\ell-k =k \cdot w -1\geq 2-1 =1 $, so there is at least one $y$-variable.
 
Initially, we compute all cliques of size $k$ in the graph:
	\begin{align*}
		T(z_1, \dots, z_k) & \obtainedfrom  V(z_1) \otimes V(z_2) \otimes \dots \otimes V(z_{k}).\\
		C(z_1, \dots, z_k) & \obtainedfrom T(z_1, \dots, z_k) \otimes \bigotimes_{1 \leq i < j \leq k} R(z_i, z_j).
	\end{align*}
	Both rules are acyclic (so their submodular width is 1) and the arity of the two intensional predicates we introduce is $k$.
	Next, we introduce the following rule:
	\begin{align*}
		S(x_1, \dots, x_k)  \obtainedfrom \bigoplus_{y_1, \ldots, y_{\ell-k}} U_1[x_1, y_1, y_2, \dots, y_{\ell-k}] \otimes  U_2[x_2, y_1, y_2, \dots, y_{\ell-k}] \otimes \dots \otimes  U_k[x_k, y_1, y_2, \dots, y_{\ell-k}].  
	\end{align*}
	where $U_i [x_i, y_1, y_2, \dots, y_{\ell-k}]$ stands for the product of the following sequence of atoms:
	\begin{align*}
	\left \{ C(z_1, \dots, z_k) \mid z_1, \dots, z_k \text{ are all possible sets of size $k$ from } \{x_i, y_1, y_2, \dots, y_{\ell-k}\} \right \} 
	\end{align*}
	Note that this is well-defined, since $\ell-k+1 = k \cdot w$ and $w \geq 1$.
	We will argue that the submodular width of this rule for $S(x_1, \dots, x_k)$ is $w$. The idea is that every tree decomposition has the exact same set of bags: $\{x_i, y_1, y_2, \dots, y_{\ell-k}\}$ for $i=1, \dots \ell$. This follows from \cite{CarmeliKKK21} and the fact that the clique-graph of the hypergraph is chordal (i.e. any cycle in the clique-graph of length greater than $3$ has a chord).
	
	Now, each bag has $\ell-k+1$ variables. Pick a set of $\ell-k+1$ hyperedges (of size $k$) such that each variable is covered exactly $k$ times (this selection is possible because every possible $k$-sized hyperedge is in the hypergraph). By assigning a weight of $1/k$ to any such hyperedge, we have a fractional edge cover of size $(\ell-k+1)/k =w$.
		
Finally, we need to check that each tuple in $S$ forms an $k$-clique, which can be done through the following rule:
	\begin{align*}
		Q() &  \obtainedfrom \bigoplus_{x_1, x_2, \ldots, x_{k}} S(x_1, \dots, x_k) \otimes  \bigotimes_{1 \leq i < j \leq k} R(x_i, x_j).  
	\end{align*}
Observe that the last rule is also acyclic, so its submodular width is 1.
\end{proof}

\subsection{Missing Proof for Part $(1)$ of Theorem~\ref{thm:main_lower_bound}}
\introparagraph{Provenance Polynomials} Consider a non-recursive \Datalogo\ program $P$ over a semiring $\bS$ with a single Boolean target $Q()$. Then, we can generate a provenance polynomial for $P$ by interpreting each \EDB\ fact $R(\ba)$ as a variable $x_{\ba}$. We call this the provenance polynomial of $P$. The non-recursive condition guarantees that the provenance polynomial is finite.

\introparagraph{Circuits over Semirings} We define a circuit $F$ over a semiring $\bS$ to be a Directed Acyclic Graph (DAG) with input nodes (with fan-in 0) variables $x_1, \dots, x_m$ and the constants $\mathbf{0}, \mathbf{1} \in \boldsymbol{D}$. Every other node is labelled by $\oplus$ or $\otimes$ and has fan-in 2; these nodes are called $\oplus$-gates and $\otimes$-gates respectively. An output gate of $F$ is any gate with fanout 0. The size of the circuit $F$, denoted $|F|$, is the number of gates in $F$. We say that $F$ {\em evaluates a polynomial} $f$ over the semiring $\bS$ if for any input values that the variables $x_1, \dots, x_m$ take over the domain $\boldsymbol{D}$ of the semiring, the output value of the circuit is the corresponding value of $f$.

\begin{proposition}\label{prop:grounding:circuit}
Let $G$ be a non-recursive grounded \Datalogo\ program over a semiring $\bS$ with a single Boolean target $Q()$. Then, there is a circuit $F$ over $\bS$ that evaluates the provenance polynomial of $G$ and $F$ has size $O(|G|)$. 
\end{proposition}

\begin{proof}
Given $G$, we will construct a semiring circuit $F$ as follows. Every \EDB\ fact $R_J(\ba_J)$ maps to a distinct input gate $z_{R_J(\ba_J)}$ of the circuit. For every $\IDB$ fact $S_J(\ba_J)$, we introduce an $\oplus$-gate $z_{S_J(\ba_J)}$. Take any rule $r$ with head $S_J(\ba_J)$ and body $T_{J_1}(\ba_{J_1}), \dots, T_{J_k}(\ba_{J_k})$. We then create a subcircuit of the form
 $$ (((z_{T_{J_1}(\ba_{J_1})} \otimes z_{T_{J_2}(\ba_{J_2})}) \otimes z_{T_{J_3}(\ba_{J_3})}) \dots )  $$
 and name the top gate of this subcircuit as $y_{S_J(\ba_J)}^r$. Let $r_1, r_2, \dots, r_s$ be the rules with the same head $S_J(\ba_J)$. Then, we make $z_{S_J(\ba_J)}$ be the top gate of the subcircuit
 $$ \left(y_{S_J(\ba_J)}^{r_1} \oplus y_{S_J(\ba_J)}^{r_2} \right) \oplus y_{S_J(\ba_J)}^{r_3} \dots $$
 The gate corresponding to $z_{q()}$ is the output gate. 
 
 We now bound the size of the circuit $F$. Note that for every rule with $k$ atoms in the body we introduce $(k-1)$ $\otimes$-gates and one $\oplus$-gate. This concludes the proof.  
\end{proof}

We will use the following proposition, which is an unconditional lower bound on the size of a semiring circuit that computes the $k$-clique.

\begin{proposition}[\cite{Jukna15}]\label{prop:clique:lb}
Any semiring circuit that evaluates the $k$-clique polynomial over $\mathsf{Trop}^{+}=\left(\mathbb{R}_{+} \cup\{\infty\}, \min ,+, \infty, 0\right)$ has size $\Omega(n^k)$.
\end{proposition}

We are now ready to prove the first part of Theorem~\ref{thm:main_lower_bound}.

\begin{proof}[Proof of Theorem~\ref{thm:main_lower_bound}, part (1)]
Consider the program $P$ from Proposition~\ref{prop:clique:program}. This program returns the minimum weight of a $k$-clique in a graph $G$, where $k = \ell-1 + \ell \cdot w$. Let $f_P$ be the provenance polynomial for $P$. From Proposition~\ref{prop:clique:lb}, we know that any semiring circuit that evaluates $f_P$ over the $\mathsf{Trop}^{+}$ has size $\Omega(n^k)$. Applying finally Proposition~\ref{prop:grounding:circuit}, we get that the size of any grounded program equivalent to $P$ over $\mathsf{Trop}^{+}$ must also be $\Omega(n^k)$.
\end{proof}

\subsection{Missing Proof for Part $(2)$ of Theorem~\ref{thm:main_lower_bound}}
\introparagraph{Conditional Lower Bounds} Using the construction of~\autoref{prop:clique:program}, we can obtain directly a conditional lower bound on the running time using problems from fine-grained complexity. In particular, we consider the {\em min-weight $\ell$-Clique hypothesis}~\cite{DBLP:conf/soda/LincolnWW18}: there is no algorithm that can find a $\ell$-Clique of minimum total edge weight in an $n$-node graph with nonnegative integer edge weights in time $O(n^{\ell-c})$ for some constant $c>0$. Observing that running the \Datalogo\ program of~\autoref{prop:clique:program} over the tropical semiring corresponds to solving the min weight $\ell$-Clique problem, we obtain the following proposition, which proves part (2) of Theorem~\ref{thm:main_lower_bound}.

\begin{proposition}
Take any integer $k \geq 2$ and any rational number $w \geq 1$ such that $k \cdot w$ is an integer. There is a (non-recursive, non-linear) \Datalogo\ program $P$ over the tropical semiring with \IDB and \EDB arities at most $k$ such that $\subw(P) = w$, such that no algorithm can compute it in time $O(n^{k w + k-1-o(1)})$ assuming the min-weight $k$-Clique hypothesis.
\end{proposition}

\section{Missing Details from Section~\ref{sec:grounding}} \label{app:grounding}

In this part, we include:

\begin{enumerate}
    \item [(\ref{sec:PANDA})] An evaluation algorithm for \SumProd queries over a dioid using the \textsf{PANDA} algorithm, formally stated in Theorem~\ref{thm:generalizedPANDA}.
    \item [(\ref{app:contexample})] A continuation of Example~\ref{ex:diamond} from Section~\ref{sec:grounding}, showing that a na\"ive grounding strategy using \textsf{PANDA} may lead to a suboptimal grounding.
    \item [(\ref{app:proofsSec6})] Missing details of Theorem~\ref{thm:main} from Section~\ref{sec:grounding}: the general grounding algorithm (Algorithm~\ref{groundedprogram}) and correctness proof.
    \item [(\ref{app:miscellaneous})] Missing details from Section~\ref{sec:grounding}: Free-connexity, Linear Programs and Fractional Hypertree-width.
\end{enumerate}

\subsection{Sum-product query evaluation using the \textsf{PANDA} algorithm} \label{sec:PANDA}

The \textbf{Proof-Assisted Entropic Degree-Aware} (\textsf{PANDA}) algorithm~\cite{PANDA} is originally introduced to evaluate CQs, i.e. \SumProd queries over the Boolean semiring. This section demonstrates that the \textsf{PANDA} algorithm can be used to evaluate \SumProd queries over any dioid. First, we introduce the concept of assignment functions for tree decompositions.

\introparagraph{Assignment Function} An \textit{assignment function} $\nu: \mE \mapsto V(\mT)$ for a tree decomposition $(\mT, \chi)$ is a function that maps each hyperedge $J \in \mE$ into a node $t$ in the decomposition such that $J \subseteq \chi(t)$. 

\medskip

Let $T(\bx_H) \obtainedfrom \sumprod_1(\bx_H) \oplus \sumprod_2(\bx_H) \oplus \ldots$ be a \Datalogo rule over a dioid $\bS$, where $\sumprod(\bx_H) : \bigoplus_{\bx_{[\ell] \setminus H}} \bigotimes_{J \in \mathcal{E}} T_J(\mathbf{x}_J)$ is a \SumProd query in the body of the rule. The $i$-th \textsf{ICO} freezes $T_J(\mathbf{x}_J)$ and evaluates $\sumprod(\bx_H)$ as a non-recursive \SumProd query over the dioid $\bS$ to get the output $\sumprod^{(i+1)}(\bx_H) $, i.e.,
\begin{equation}\label{eq:faq:freeze}
    \sumprod^{(i+1)}(\bx_H) : \quad \bigoplus_{\bx_{[\ell] \setminus H}} \bigotimes_{J \in \mathcal{E}} T_J^{(i)}(\mathbf{x}_J).
\end{equation}
\textsf{PANDA} evaluates \eqref{eq:faq:freeze} when $\bS$ is the Boolean semiring in time $\widetilde{O}(m^{\subw} \cdot |\sumprod^{(i+1)}(\bx_H) |)$, assuming $m$ is the sum of sizes of the input predicates $T_J^{(i)}(\mathbf{x}_J)$ and $\sumprod^{(i+1)}(\bx_H)$ is the output (or derived) $\bS$-relation.

\medskip

Now we present the evaluation algorithm for \SumProd query evaluation over any \textit{dioid} $\bS$. The algorithm is as follows:
\begin{enumerate}
    \item [(1)] Take every non-redundant tree decomposition of the hypergraph $([\ell], \mE)$ associated with the \SumProd query $\sumprod(\bx_H)$
    $$ (\mT_1, \chi_1),  (\mT_2, \chi_2), \ldots, (\mT_M, \chi_M)
    $$
    We assume that there are $M$ of them, where $M$ is finite under data complexity, i.e. independent of the input instance (Proposition 2.9 in \cite{PANDA}).
    \item [(2)] Write a set of disjunctive rules (where the head is a disjunction of output relations)
    \begin{align}\label{eq:disjunctiveRule}
         \rho: \bigvee_{i \in [M]} B_{\chi_i(t_i)}(\bx_{\chi_i(t_i)}) \leftarrow \bigwedge_{J \in \mathcal{E}} B_J(\mathbf{x}_J) 
    \end{align}
     where $B_J(\bx_J)$ are $\mathbb{B}$-relations and a tuple $\ba_J \in B_J(\bx_J)$ if and only if $T_J(\ba_J) \neq \zerobf$; and each disjunctive rule $\rho$ picks one node $t_i \in V(\mT_i)$ from every $(\mT_i, \chi_i)$ to construct its head of $\rho$. Thus, there are $\Pi_{i \in [M]} |V(\mT_i)|$ such disjunctive rules. 
    \item [(3)] On each disjunctive rule $\rho$ that selected nodes $(t_1, t_2, \cdots, t_{M}) \in \times_{i \in [M]} V(\mathcal{T}_i)$, apply the $\mathsf{PANDA}$ algorithm to get a materialized view $B^*_{\chi_i(t_i)}(\bx_{\chi_i(t_i)})$ for every ${i \in [M]}$. It is shown by~\cite{PANDA} that \textsf{PANDA} evaluates any disjunctive rule in \eqref{eq:disjunctiveRule} in $\widetilde{O}(m^{\subw})$ time and every output view $B^*_{\chi_i(t_i)}$ is of size $\widetilde{O}(m^{\subw})$. The reader may refer to~\cite{PANDA} for the formal evaluation semantics of a disjunctive rule.
    \item [(4)] For each node $t \in V(\mT_i)$ in each $(\mT_i, \chi_i)$, evaluate the \SumProd query
    $$ \overline{T_{\chi_i(t)}}(\bx_{\chi_i(t)}) \obtainedfrom \overline{B_{\chi_i(t)}}(\bx_{\chi_i(t)}) \otimes  \bigotimes_{J \in \mathcal{E} : \nu_i(J) = t} T_J(\bx_J),
    $$
    where the input $\bS$-relation $\overline{B_{\chi_i(t)}}(\bx_{\chi_i(t)})$ stores a tuple $\bx_{\chi_i(t)}$ as value $\onebf$ if the tuple is in the union over all $B^*_{\chi_i(t)}(\bx_{\chi_i(t)})$ emitted by (possibly multiple) \textsf{PANDA} invocations in step (3) across all disjunctive rules. Indeed, observe that each node $t \in V(\mT_i)$ participates in $\Pi_{j \in [M] \setminus i} |V(\mT_j)|$ disjunctive rules. Here, let $\nu_i : \mathcal{E} \rightarrow V(\mT_i)$ be an arbitary assignment function for $(\mT_i, \chi_i)$ that assigns each $J \in \mathcal{E}$ to a node $t \in V(\mT_i)$ such that $J \subseteq \chi_i(t)$. Such a $\nu_i$ is always possible by the definition of a tree decomposition.
    \item [(5)] For every $(\mT_i, \chi_i)$, evaluate the acyclic \SumProd query 
    \begin{equation} \label{eq:intermediateFAQ}
         \sumprod_i(\bx_H) \obtainedfrom \bigoplus_{\bx_{[\ell] \setminus H}} \bigotimes_{t \in V(\mT_i)} \overline{T_{\chi_i(t)}}(\bx_{\chi_i(t_i)}) 
    \end{equation}
    using the \textsf{InsideOut} algorithm~\cite{faq} in time $\widetilde{O}(|\overline{T_{\chi_i(t)}}(\bx_{\chi_i(t_i)})| \cdot |\sumprod_i(\bx_H)|)$, where $|\overline{T_{\chi_i(t)}}(\bx_{\chi_i(t_i)})|$ is the size of input relation.
    \item [(5)] Get the final query result $\sumprod(\bx_H)$ by $\bigoplus_{i \in [M]} \sumprod_i(\bx_H)$.
\end{enumerate}

The correctness of the algorithm over the Boolean semiring and when $H = [\ell]$ or $\emptyset$ is established in Corollary 7.13 of~\cite{PANDA}. We lift the corollary to a dioid, and over arbitrary $\emptyset \subseteq H \subseteq [\ell]$.

\begin{lemma} \label{lemma:generalizedPANDA}
    $\sumprod(\bx_H) = \bigoplus_{i \in [M]} \sumprod_i(\bx_H)$, where $\sumprod(\bx_H) $ is the target \SumProd query \eqref{eq:faq:freeze} and $\sumprod_i(\bx_H)$ is an acyclic $query$ \eqref{eq:intermediateFAQ} from the $i$-th tree decomposition $(\mT_i, \chi_i)$.
\end{lemma}
\begin{proof}
    By Corollary 7.13 of~\cite{PANDA} (the correctness over the Boolean semiring), for any constant full tuple $\ba_{[\ell]}$, $\sumprod(\ba_{[\ell]}) := \bigotimes_{J \in \mathcal{E}} T_J(\ba_J) \neq \zerobf$ if and only if there is at least one decomposition $(\mT_i, \chi_i)$, where $i \in [M]$, such that $\sumprod_i(\ba_{[\ell]}) := \bigotimes_{t \in V(\mT_i)} \overline{T_{\chi_i(t)}}(\ba_{\chi_i(t)})  \neq \zerobf$. Here, we use $\ba_J$ as a shorthand for the projection of $\ba_{[\ell]}$ on the variables in $\bx_J$. In fact, for any such $i$  where $\phi_i(\ba_{[\ell]}) := \bigotimes_{t \in V(\mT_i)} \overline{T_{\chi_i(t)}}(\ba_{\chi_i(t)}) \neq \zerobf$ (possibly many of them), 
    \begin{align*}
        \phi_i(\ba_{[\ell]})  & = \bigotimes_{t \in V(\mT_i)} \overline{T_{\chi_i(t)}}(\ba_{\chi_i(t)})  \\
                        & = \bigotimes_{t \in V(\mT_i)} \overline{B_{\chi_i(t)}}(\ba_{\chi_i(t)}) \otimes  \bigotimes_{J \in \mathcal{E} : \nu_i(J) = t} T_J(\ba_J) \\
                        & = \bigotimes_{t \in V(\mT_i)} \onebf  \otimes  \bigotimes_{J \in \mathcal{E} : \nu_i(J) = t} T_J(\ba_J) \\
                        & = \bigotimes_{J \in \mathcal{E}} T_J(\ba_J) 
    \end{align*}
    where the third line uses the fact that $\overline{B_{\chi_i(t)}}(\ba_{\chi_i(t)}) = \onebf$ because otherwise, $\phi_i(\ba_{[\ell]}) = \zerobf$ (a contradiction). Thus, by idempotence of $\oplus$, $\bigoplus_{i \in [M]}  \phi_i(\ba_{[\ell]}) =  \bigotimes_{J \in \mathcal{E}} T_J(\ba_J) $ for any $\ba_{[\ell]}$. Now for an arbitary tuple $\ba_H$,
    \begin{align*}
        \bigoplus_{i \in [M]} \sumprod_i(\ba_H) & = \bigoplus_{i \in [M]} \bigoplus_{\bx_{[\ell] \setminus H}} \phi_i(\ba_{H}, \bx_{[\ell] \setminus H}) \\
                                         & = \bigoplus_{\bx_{[\ell] \setminus H}}  \bigoplus_{i \in [M]}  \phi_i(\ba_{H}, \bx_{[\ell] \setminus H}) & \textit{(change aggregation order)} \\
                                         & = \bigoplus_{\bx_{[\ell] \setminus H}} \left( \bigotimes_{J \in \mathcal{E}} T_J(\Pi_J(\ba_{H}, \bx_{[\ell] \setminus H})) \right)  & \textit{(by idempotence of $\oplus$)} \\ 
                                         & = \sumprod(\ba_{ H})
    \end{align*}
    where $(\ba_{H}, \bx_{[\ell] \setminus H})$ denotes an arbitary (full) tuple over the active domain of $\bx_{[\ell]}$ where its projection on $\bx_H$ agrees with $\ba_H$. 
\end{proof}

As a result of the Lemma~\ref{lemma:generalizedPANDA} and Theorem 1.7 of~\cite{PANDA}, we have the following theorem.
\begin{theorem} \label{thm:generalizedPANDA}
    The above evaluation algorithm evaluates a \SumProd query over a dioid $\bS$ in time $\widetilde{O}(m^{\subw} \cdot |\sumprod(\bx_H)|)$, where $\subw$ is the submodular width of the query,  $m$ is the sum of sizes of the input $\bS$-relations $T_J(\mathbf{x}_J)$ and $\sumprod(\bx_H)$ is the output $\bS$-relation.
\end{theorem}

\smallskip

\subsection{Continuing Example~\ref{ex:diamond}} \label{app:contexample}
\begin{example} \label{app:ex:diamond}
    In Example~\ref{ex:diamond}, we ground the following \Datalogo program $P$ that finds nodes that are reachable via a dimond-pattern from some node in a set $U$:
    \begin{align*}
     T(x_1) &\obtainedfrom U(x_1) \oplus  \bigoplus_{\bx_{234}} T(x_3) \otimes R_{32}(x_3, x_2) \otimes R_{21}(x_2, x_1) \otimes R_{34}(x_3, x_4) \otimes R_{41}(x_4, x_1) 
    \end{align*}
    A simple grounding is to ground the $4$-cycle join first, materializing the opposite nodes of the cycle, i.e.
    \begin{align*}
        \textsf{Cycle}(x_1, x_3) & \obtainedfrom \bigoplus_{\bx_{24}} R_{32}(x_3, x_2) \otimes R_{21}(x_2, x_1) \otimes R_{34}(x_3, x_4) \otimes R_{41}(x_4, x_1) \\
        T(x_1) &\obtainedfrom U(x_1) \oplus  \bigoplus_{\bx_{3}} T(x_3) \otimes \textsf{Cycle}(x_1, x_3) 
       \end{align*}
    \textsf{PANDA} evaluates the first rule (essentially a non-recursive \SumProd query) in time $\widetilde{O}(m^{3/2} + n^2)$, since $\textsf{Cycle}(x_1, x_3)$ has a grounding (or output) of size $O(n^2)$.  
\end{example}

\subsection{Missing Details of Theorem~\ref{thm:main} from Section~\ref{sec:grounding}: Formal Algorithm and Proof} \label{app:proofsSec6}

Algorithm~\ref{groundedprogram} is a rule-by-rule construction of the $\bS$-equivalent grounding $G$ from $P$ and it leads to Theorem~\ref{thm:main}. Its formal proof is shown as follows.

 \begin{algorithm}[!t]
	\SetKwInOut{Input}{Input}
	\SetKwInOut{Output}{Output}
	\SetKwFunction{ground}{grounds}
	\SetKwFunction{rewrite}{rewrite}
	\SetKwFunction{genprogram}{GenerateProgram}
	\SetKwProg{myproc}{Procedure}{}{}
	\Input{Datalog Program $P$ and dioid $\sigma$}
	\Output{$\bS$-equivalent Grounding $G$}
        \textbf{let} $G = \emptyset$ \\
        \ForEach{rule $T(\bx_H) \obtainedfrom \sumprod_1(\bx_H) \oplus \sumprod_2(\bx_H) \ldots \in P$} {
            \ForEach{\SumProd query $\sumprod(\bx_H)$ in the rule}{ 
                \textcolor{gray}{// \textit{assume} $\sumprod(\bx_H) =  \bigoplus_{\bx_{[\ell] \setminus H}} \bigotimes_{J \in \mathcal{E}} T_J(\bx_J)$} \\
                \textcolor{gray}{// $(i)$ \textit{Refactoring over multiple tree decompositions}} \\

                \textbf{construct} the set of all non-redundant tree decompositions of $\sumprod(\bx_H)$, i.e. $(\mT_1, \chi_1),  (\mT_2, \chi_2), \ldots, (\mT_M, \chi_M)$; \\
                \textbf{associate} each $(\mT_i, \chi_i)$ with an assignment function $\nu_i : \mathcal{E} \rightarrow V(\mT_i)$; \\

                \textcolor{gray}{// $(ii)$ \textit{Grounding all bags/IDBs at once using \textsf{PANDA}} (see Appendix~\ref{sec:PANDA} and~\cite{PANDA} for details of the \textsf{PANDA} algorithm)} \\
                \textbf{introduce} a new \IDB $B_{\chi_i(t)}(\bx_{\chi_i(t)})$ for every $t \in V(\mT_i)$ and $i \in [M]$; \\
                \textbf{apply} \textsf{PANDA} (more precisely step (2), (3) and (4) as illustrated in Appendix~\ref{sec:PANDA}) to get a view $\overline{B_{\chi_i(t)}}(\bx_{\chi_i(t)})$ for every $t \in V(\mT_i)$ and $i \in [M]$; \\
                \ForEach{tuple $\ba_{\chi_i(t)} \in \overline{B_{\chi_i(t)}}(\bx_{\chi_i(t)})$}{ 
                    \textbf{insert} into $G$ the grounded rule $B_{\chi_i(t)}(\ba_{\chi_i(t)}) \obtainedfrom \bigotimes_{J \in \mathcal{E} : \nu_i(J) = t} T_J(\Pi_{J}(\ba_{\chi_i(t)}))$; \\
                }
                \textcolor{gray}{// $(iii)$ \textit{Grounding the acyclic sub-queries, one for each decomposition}} \\
                \textbf{apply} Algorithm~\ref{grounded-acyclicprogram} to ground the acyclic rule 
                $T(\bx_H) \obtainedfrom \bigoplus_{\bx_{i \in [M]}} \left( \bigoplus_{\bx_{[\ell] \setminus H}} \bigotimes_{t \in V(\mT_i)} B_{\chi_i(t)}\right)$, \\
                using $\{B^*_{\chi_i(t)}(\bx_{\chi_i(t_i)}) \mid t \in V(\mT_i)\}$ as the input instance; \\
            }
        }
	\caption{Ground a \Datalogo Program over a dioid $\bS$} \label{groundedprogram}
\end{algorithm}

\begin{proof}[Proof of Theorem~\ref{thm:main}]
    The key idea of Algorithm~\ref{groundedprogram} is to group atoms in a (possibly cyclic) \SumProd query $\sumprod(\bx_H)$ into an acyclic one using tree decompositions, then ground the acyclic \SumProd query using Algorithm~\ref{grounded-acyclicprogram}. Suppose that $\sumprod(\bx_H)$ has a tree decomposition $(\mT, \chi)$ with an assignment function $\nu : \mathcal{E} \rightarrow V(\mT)$. Then, we can rewrite $\sumprod(\bx_H)$ as
    \begin{align*}
        \sumprod(\bx_H) & = \bigoplus_{\bx_{[\ell] \setminus H}} \bigotimes_{J \in \mathcal{E}} T_J(\bx_J)  = \bigoplus_{\bx_{[\ell] \setminus H}} \bigotimes_{t \in V(\mT)} \bigotimes_{J \in \mE: \nu(J) = t} T_{J}(\bx_{J}) = \bigoplus_{\bx_{[\ell] \setminus H}} \bigotimes_{t \in V(\mT)} B_{\chi(t)}(\bx_{\chi(t)})
    \end{align*}
    where we substitute $B_{\chi(t)}(\bx_{\chi(t)})$ by $\bigotimes_{J \in \mE: \nu(J) = t} T_{J}(\bx_{J})$. The resulting \SumProd query is acyclic because $\mT$ is its join tree.  If using multiple (say $M$) tree decompositions as in Algorithm~\ref{groundedprogram}, we can rewrite $\sumprod(\bx_H)$ as
    \begin{align*}
        \sumprod(\bx_H) & = \bigoplus_{\bx_{[\ell] \setminus H}} \bigotimes_{J \in \mathcal{E}} T_J(\bx_J) \\
        & = \bigoplus_{i \in [M]} \bigoplus_{\bx_{[\ell] \setminus H}} \bigotimes_{J \in \mathcal{E}} T_J(\bx_J) \\
        & = \bigoplus_{i \in [M]} \bigoplus_{\bx_{[\ell] \setminus H}} \bigotimes_{t \in V(\mT_i)} B_{\chi_i(t)}(\bx_{\chi_i(t)}) 
    \end{align*}
    where the second equality uses the idempotency of $\oplus$. To ground the new \IDB $B_{\chi_i(t)}(\bx_{\chi_i(t)})$, we apply the \textsf{PANDA} algorithm over the dioid $\bS$ to get a view $B^*_{\chi_i(t)}(\bx_{\chi_i(t)})$, for every $t \in V(\mT_i)$ and $i \in [M]$. By Lemma~\ref{lemma:generalizedPANDA} in Appendix~\ref{sec:PANDA}, for any tuple $\bx_H$, it always holds that
    $$
    \sumprod(\bx_H) = \bigoplus_{i \in [M]} \bigoplus_{\bx_{[\ell] \setminus H}}  \bigotimes_{t \in V(\mT_i)} B^*_{\chi_i(t)}(\bx_{\chi_i(t)}).
    $$ 
    Thus, grounding each bag $B_{\chi_i(t)}(\bx_{\chi_i(t)})$ by its corresponding view $B^*_{\chi_i(t)}(\bx_{\chi_i(t)})$ suffices, i.e. we insert into $G$ a rule:
    $$
        B_{\chi_i(t)}(\ba_{\chi_i(t)}) \obtainedfrom \bigotimes_{J \in \mathcal{E} : \nu_i(J) = t} T_J(\Pi_{J}(\ba_{\chi_i(t)}))
    $$
    for every $\ba_{\chi_i(t)} \in B^*_{\chi_i(t)}(\bx_{\chi_i(t)})$. Finally, to ground $T(\bx_H)$, we simply call Algorithm~\ref{grounded-acyclicprogram} and since the size of $B^*_{\chi_i(t)}(\bx_{\chi_i(t)})$ is $\widetilde{O}(m^\subw  + n^{k \cdot \subw})$, by Theorem~\ref{thm:main:acyclic}, we can construct a grounding for $T(\bx_H)$ in time $\widetilde{O}(n^{k-1} \cdot (m^\subw  + n^{k \cdot \subw}))$ of size $\widetilde{O}(n^{k-1} \cdot (m^\subw  + n^{k \cdot \subw}))$.
\end{proof}

\subsection{Missing details from Section~\ref{sec:grounding}: Free-connexity, Linear Programs and Fractional Hypertree-width.} \label{app:miscellaneous}


\begin{proof}[Proof of Proposition~\ref{thm:main:free-connex}]
    The proof follows directly from that of Theorem~\ref{thm:main} and Theorem~\ref{thm:main:acyclic:freeconnex}, by constraining Algorithm~\ref{groundedprogram} to only use non-redundant free-connex tree decompositions.
\end{proof}

\smallskip

\introparagraph{Degree Constraints} 
In fact, beyond the input size $m$, the evaluation algorithm presented in Appendix~\ref{sec:PANDA} can incorporate a broader class of constraints called \textit{degree constraints}~\cite{PANDA} to evaluate a \SumProd query over a dioid more efficiently.
 A degree constraint~\cite{PANDA} is an assertion on an atom $T_J(\bx_J)$ and $V \subset V \subseteq J$ of the form 
\begin{align} \label{eq:degreeConstraint}
    \textsf{deg}(\bx_U \mid \bx_V) = \max_{\ba_V \in \mathsf{Dom}(\bx_V)} | \Pi_{U} T(\bx_{J \setminus V}, \ba_V ) | \leq m_{U \mid V},
\end{align}
where $\Pi_{U} T(\bx_{J \setminus V}, \ba_V )$ is the relation obtained by first selecting tuples that agree with the tuple $\ba_V$ on $\bx_V$ and then projecting on $\bx_U$ and $m_{U \mid V}$ is a natural number. A cardinality constraint in our setting is therefore a degree constraint with $m_{J \mid \emptyset} = m$. Similarly, the active domain of $n$ for a variable $x_i$ ($i \in J$) is a degree constraint with $m_{\{i\} \mid \emptyset} = n$. In fact, Proposition~\ref{prop:main:linear} is an improved grounding results for linear \Datalog, using the constraints of active domains on each variable. 

In addition, a functional dependency $V \rightarrow U$ is a degree constraint with $m_{U \cup V \mid V} = 1$. This can capture primary-key constraints \EDBs, and repeated variables in \IDBs (e.g. $T(x,x,y)$ essentially says, the first variable decides the second).

The evaluation algorithm presented in Appendix~\ref{sec:PANDA} uses \textsf{PANDA} as a black-box, thus it naturally inherits the capability of handling degree constraints as \eqref{eq:degreeConstraint} from the original \textsf{PANDA} algorithm, to evaluate a \SumProd query $\sumprod(\bx_H)$ over a dioid more efficiently. Its runtime is shown to be predicated by the {\em degree-aware submodular width}. We refer the reader to~\cite{PANDA} for the formal definition of degree-aware submodular width and in-depth analysis.

\medskip 

Using \textsf{PANDA} as a subroutine, our grounding algorithm (Algorithm~\ref{groundedprogram}) also handles degree constraints for tighter groundings. For example, we show an improved grounding for linear \Datalogo, using the refined constraint that each variable in an \IDB has an active domain of $n$ (see Table~\ref{tab:summary}), instead of $O(n^k)$ for every \IDB predicate.

\begin{proof}[Proof of Proposition~\ref{prop:main:linear}]
    We follow the exact proof of Theorem~\ref{thm:main}. However, instead of the crude $\widetilde{O}(m^{\subw} + n^{k \cdot \subw})$ bound for each view $B^*_{\chi_i(t)}(\bx_{\chi_i(t)})$ obtained at Line 9 of Algorithm~\ref{groundedprogram}, we show a tighter bound of $\widetilde{O}(m^{\subw} \cdot (n^k / m + 1))$. 
    
    By the properties of \textsf{PANDA}~\cite{PANDA} and the fact that all rules are linear, the fine-grained cardinality constraints are: $(i)$ at most one \IDB in each \SumProd query with an active domain of $O(n)$ for each of its variable, $(ii)$ other \EDB atoms are of size $O(m)$. Note that $(i)$ is a stronger constraint than the \IDB cardinality constraint of $O(n^k)$.

    We can assume that \textsf{PANDA} terminates in time $\widetilde{O}(m^{w-w_i} \cdot n^{k w_i } )$ and the size of each resulting view $B^*_{\chi_i(t)}(\bx_{\chi_i(t)})$ has a more accurate size bound of $\widetilde{O}(m^{w-w_i} \cdot n^{k w_i } ) = \widetilde{O}(m^w \cdot (n^k/m)^{w_i})$, where $w \leq \subw$ and $w_i \leq 1$. One may refer to Lemma 5.3 of~\cite{PANDA} for the primal and dual linear programs that justifies the assumed bound.
    
    We now split into two cases: if $n^k \leq m$, then $(n^k/m)^{w_i} \leq 1$; otherwise,  $(n^k/m)^{w_i} \leq n^k/m$. Hence, we get $\widetilde{O}(m^w \cdot (n^k/m +1 ))$, where $w \leq \subw$ for groundings intermediate bags. 
\end{proof}

\medskip 

\begin{proof}[Proof of Proposition~\ref{prop:main:fhw}]
    We still use Algorithm~\ref{groundedprogram} to ground $P$, with two slight modifications: $(i)$ for each \SumProd query $\sumprod$, pick one tree decomposition $(\mT, \chi)$ that attains $\fhw(\sumprod)$; and $(ii)$ use the \textsf{InsideOut} algorithm instead of \textsf{PANDA} on when grounding the bags $B_{\chi(t)}(\bx_{\chi(t)})$. Thus, the size of the grounding for each bag is $O(m^{\fhw} + n^{k \cdot \fhw})$, without the polylog factor! The rest of the proof is similar to that of Theorem~\ref{thm:main} without using the idempotence of $\oplus$, since only one tree decomposition is used for each \SumProd query ($M = 1$ in that proof).
\end{proof}

Similar to the free-connex submodular-width, the free-connex fractional hypertree-width of a \SumProd query $\sumprod(\bx_H)$ is defined as 
\begin{align*}
    \cfhw(\sumprod(\bx_H)) \defeq \; \min_{(\mathcal{T}, \chi) \in \mathcal{F}_f}  \max_{h \in \Gamma_{\ell}} \max_{t \in V(\mathcal{T})} h(\chi(t)) 
\end{align*}
where $\mathcal{F}_f$ is the set of all non-redundant free-connex tree decompositions of $\sumprod(\bx_H)$. Obviously, $\cfhw(\sumprod(\bx_H)) \leq \fhw(\sumprod(\bx_H))$. We obtain the following corollary directly from Theorem~\ref{thm:main:acyclic:freeconnex} and Proposition~\ref{prop:main:fhw}.

\begin{corollary} \label{cor:main:cfhw}
    Let $P$ be a \Datalogo program where $\arity{P}$ is at most $k$. Let $\cfhw$ be the free-connex fractional hypertree-width of $P$.
    Let $\bS$ be a naturally-ordered semiring. Then, a $\bS$-equivalent grounding can be constructed of size (and time) $O(m^{\cfhw} + n^{k \cdot \cfhw})$.
\end{corollary}

\end{document}